\PassOptionsToPackage{dvipsnames}{xcolor}
\documentclass[11pt,letterpaper,twoside]{article}
\usepackage{fullpage}
\usepackage{booktabs}
\usepackage{amsmath}
\usepackage{amsfonts}
\usepackage{amssymb}
\usepackage{amsthm}
\usepackage{upgreek}
\usepackage{times}

\usepackage{color}
\usepackage[final]{graphicx}
\usepackage{tikz}
\usepackage{algorithmic, algorithm}

\def\P{\mathbb{P}}

\def\R{\mathbb{R}}

\def\1{\mathbf{1}}

\def\stocmode{0}

\def\arxivmode{0}
\def\fastmode{0}

\def\showauthornotes{0}

\def\showkeys{0}
\def\showdraftbox{1}
\def\showcolorlinks{0}
\def\usemicrotype{1}
\def\showfixme{1}

\ifnum\fastmode=0
\usepackage{etex}
\fi

\ifnum\fastmode=0
\usepackage[l2tabu, orthodox]{nag}
\fi

\ifnum\fastmode=0
\usepackage[T1]{fontenc}
\usepackage[full]{textcomp}
\fi

\usepackage[american]{babel}

\usepackage{mathtools}

\usepackage{amsthm}

\newtheorem{theorem}{Theorem}[section]
\newtheorem*{theorem*}{Theorem}

\newtheorem{proposition}[theorem]{Proposition}
\newtheorem*{proposition*}{Proposition}
\newtheorem{lemma}[theorem]{Lemma}
\newtheorem*{lemma*}{Lemma}
\newtheorem{corollary}[theorem]{Corollary}
\newtheorem*{conjecture*}{Conjecture}

\newtheorem*{fact*}{Fact}

\newtheorem*{exercise*}{Exercise}

\newtheorem*{hypothesis*}{Hypothesis}
\newtheorem{conjecture}[theorem]{Conjecture}

\theoremstyle{definition}
\newtheorem{definition}[theorem]{Definition}

\newtheorem{exercise-easy}[theorem]{Exercise}
\newtheorem{exercise-med}[theorem]{Exercise}
\newtheorem{exercise-hard}[theorem]{Exercise$^\star$}

\newtheorem*{claim*}{Claim}

\newtheorem*{remark*}{Remark}

\newtheorem*{observation*}{Observation}

\usepackage[
letterpaper,
top=0.65in,
bottom=0.7in,
left=1in,
right=1in]{geometry}

\ifnum\arxivmode=0
\usepackage{newpxtext} 
\usepackage{textcomp} 
\usepackage[varg,bigdelims]{newpxmath}
\usepackage[scr=rsfso]{mathalfa}
\usepackage{bm} 
\let\mathbb\varmathbb
\fi

\ifnum\arxivmode=1
\usepackage[varg]{pxfonts} 
\usepackage{textcomp} 
\usepackage[scr=rsfso]{mathalfa}
\usepackage{bm} 
\fi

\ifnum\showkeys=1
\usepackage[color]{showkeys}
\fi

\makeatletter
\@ifpackageloaded{hyperref}
{}
{\usepackage{hyperref}}

\definecolor{bleudefrance}{rgb}{0.01, 0.1, 1.0}
\definecolor{azure}{rgb}{0.0, 0.5, 1.0}

\ifnum\showcolorlinks=1
\hypersetup{
pagebackref=false,
colorlinks=true,
urlcolor=blue,
linkcolor=blue,
citecolor=OliveGreen}
\fi

\ifnum\showcolorlinks=0
\hypersetup{
pagebackref=false,
colorlinks=false,
pdfborder={0 0 0}}
\fi

\usepackage{prettyref}

\newcommand{\savehyperref}[2]{\texorpdfstring{\hyperref[#1]{#2}}{#2}}

\newrefformat{eq}{\savehyperref{#1}{\textup{(\ref*{#1})}}}
\newrefformat{lem}{\savehyperref{#1}{Lemma~\ref*{#1}}}
\newrefformat{def}{\savehyperref{#1}{Definition~\ref*{#1}}}
\newrefformat{thm}{\savehyperref{#1}{Theorem~\ref*{#1}}}
\newrefformat{cor}{\savehyperref{#1}{Corollary~\ref*{#1}}}
\newrefformat{cha}{\savehyperref{#1}{Chapter~\ref*{#1}}}
\newrefformat{sec}{\savehyperref{#1}{Section~\ref*{#1}}}
\newrefformat{app}{\savehyperref{#1}{Appendix~\ref*{#1}}}
\newrefformat{tab}{\savehyperref{#1}{Table~\ref*{#1}}}
\newrefformat{fig}{\savehyperref{#1}{Figure~\ref*{#1}}}
\newrefformat{hyp}{\savehyperref{#1}{Hypothesis~\ref*{#1}}}
\newrefformat{alg}{\savehyperref{#1}{Algorithm~\ref*{#1}}}
\newrefformat{rem}{\savehyperref{#1}{Remark~\ref*{#1}}}
\newrefformat{item}{\savehyperref{#1}{Item~\ref*{#1}}}
\newrefformat{step}{\savehyperref{#1}{step~\ref*{#1}}}
\newrefformat{conj}{\savehyperref{#1}{Conjecture~\ref*{#1}}}
\newrefformat{fact}{\savehyperref{#1}{Fact~\ref*{#1}}}
\newrefformat{prop}{\savehyperref{#1}{Proposition~\ref*{#1}}}
\newrefformat{prob}{\savehyperref{#1}{Problem~\ref*{#1}}}
\newrefformat{claim}{\savehyperref{#1}{Claim~\ref*{#1}}}
\newrefformat{relax}{\savehyperref{#1}{Relaxation~\ref*{#1}}}
\newrefformat{red}{\savehyperref{#1}{Reduction~\ref*{#1}}}
\newrefformat{part}{\savehyperref{#1}{Part~\ref*{#1}}}
\newrefformat{ex}{\savehyperref{#1}{Exercise~\ref*{#1}}}
\newrefformat{property}{\savehyperref{#1}{Property~\ref*{#1}}}

\newcommand{\Sref}[1]{\hyperref[#1]{\S\ref*{#1}}}

\usepackage{nicefrac}

\ifnum\fastmode=0
\ifnum\usemicrotype=1
\usepackage{microtype}
\fi
\fi

\ifnum\showauthornotes=2
\newcommand{\mynotes}[1]{{\sffamily\small\color{teal}{#1}}\medskip}
\newcommand{\Authornote}[2]{{\sffamily\small\color{blue}{[#1: #2]}}\medskip}
\newcommand{\Authornotecolored}[3]{{\sffamily\small\color{#1}{[#2: #3]}}}
\newcommand{\Authorcomment}[2]{{\sffamily\small\color{gray}{[#1: #2]}}}
\newcommand{\Authorstartcomment}[1]{\sffamily\small\color{gray}[#1: }

\newcommand{\Authorfnote}[2]{\footnote{\color{red}{#1: #2}}}
\newcommand{\Authorfixme}[1]{\Authornote{#1}{\textbf{??}}}
\newcommand{\Authormarginmark}[1]{\marginpar{\textcolor{red}{\fbox{\Large #1:!}}}}
\newcommand{\myexplain}[1]{{\sffamily\small\color{red}{\noindent [Explanation:\medskip\newline \begin{quote}#1\hfill]\end{quote}}}\medskip}

\else

\newcommand{\mynotes}[1]{}
\newcommand{\Authornote}[2]{}
\newcommand{\Authornotecolored}[3]{}
\newcommand{\Authorcomment}[2]{}
\newcommand{\Authorstartcomment}[1]{}

\newcommand{\Authorfnote}[2]{}
\newcommand{\Authorfixme}[1]{}
\newcommand{\Authormarginmark}[1]{}
\newcommand{\myexplain}[1]{}

\ifnum\showauthornotes=1
\renewcommand{\myexplain}[1]{{\sffamily\small\color{red}{\noindent \begin{quote}{\bf Explanation:} \medskip\newline #1\end{quote}}}\medskip}
\fi

\fi

\ifnum\showfixme=0

\fi

\usepackage{boxedminipage}

\newcommand{\Esymb}{\mathbb{E}}

\DeclareMathOperator*{\E}{\Esymb}

\newcommand{\textparen}[1]{\text{(#1)}}

\ifx\because\undefined
\newcommand{\because}[1]{\textparen{because #1}}
\else
\renewcommand{\because}[1]{\textparen{because #1}}
\fi

\newcommand\bdot\bullet

\ifx\mathds\undefined 

\else

\fi

\renewcommand{\leq}{\leqslant}

\renewcommand{\geq}{\geqslant}

\ifnum\showdraftbox=1

\else

\fi

\let\epsilon=\varepsilon

\numberwithin{equation}{section}

\newcommand\MYcurrentlabel{xxx}

\newcommand{\MYstore}[2]{%
  \global\expandafter \def \csname MYMEMORY #1 \endcsname{#2}%
}

\newcommand{\MYload}[1]{%
  \csname MYMEMORY #1 \endcsname%
}

\newcommand{\MYnewlabel}[1]{%
  \renewcommand\MYcurrentlabel{#1}%
  \MYoldlabel{#1}%
}

\newcommand{\MYdummylabel}[1]{}

\newcommand{\torestate}[1]{%
  \let\MYoldlabel\label%
  \let\label\MYnewlabel%
  #1%
  \MYstore{\MYcurrentlabel}{#1}%
  \let\label\MYoldlabel%
}

\newcommand{\restatetheorem}[1]{%
  \let\MYoldlabel\label
  \let\label\MYdummylabel
  \begin{theorem*}[Restatement of \prettyref{#1}]
    \MYload{#1}
  \end{theorem*}
  \let\label\MYoldlabel
}

\newcommand{\restatelemma}[1]{%
  \let\MYoldlabel\label
  \let\label\MYdummylabel
  \begin{lemma*}[Restatement of \prettyref{#1}]
    \MYload{#1}
  \end{lemma*}
  \let\label\MYoldlabel
}

\newcommand{\restateprop}[1]{%
  \let\MYoldlabel\label
  \let\label\MYdummylabel
  \begin{proposition*}[Restatement of \prettyref{#1}]
    \MYload{#1}
  \end{proposition*}
  \let\label\MYoldlabel
}

\newcommand{\restatefact}[1]{%
  \let\MYoldlabel\label
  \let\label\MYdummylabel
  \begin{fact*}[Restatement of \prettyref{#1}]
    \MYload{#1}
  \end{fact*}
  \let\label\MYoldlabel
}

\newcommand{\restate}[1]{%
  \let\MYoldlabel\label
  \let\label\MYdummylabel
  \MYload{#1}
  \let\label\MYoldlabel
}

\newcommand{\addreferencesection}{
  \phantomsection
\ifnum\stocmode=0
  \addcontentsline{toc}{section}{References}
\else
  \addcontentsline{toc}{section}{References \hspace*{1in} --------- End of extended abstract ---------}
\fi

}

\let\origparagraph\paragraph
\renewcommand{\paragraph}[1]{\vspace*{-10pt}\origparagraph{#1.}}

\allowdisplaybreaks

\usepackage{paralist}

\usepackage{comment}

\usepackage{enumitem}

\begin{document}
\title{Multi-Resolution Hashing for Fast Pairwise Summations
}
\author{Moses Charikar\\Department of Computer Science\\Stanford University \\\texttt{moses@cs.stanford.edu} \and Paris Siminelakis\\Department of Electrical Engineering\\
Stanford University\\\texttt{psimin@stanford.edu}}
\maketitle
\thispagestyle{empty}
\begin{abstract}
A basic computational primitive in the analysis of massive datasets is summing simple functions over a large number of objects. Modern applications pose an additional challenge in that such functions often depend on a parameter vector $y$ (query) that is unknown a priori.  Given a set of points $X\subset  \R^{d}$ and a pairwise function $w:\R^{d}\times \R^{d}\to [0,1]$, we study the problem of designing a data-structure that enables sublinear-time approximation of the summation  $Z_{w}(y)=\frac{1}{|X|}\sum_{x\in X}w(x,y)$ for any query $y\in \R^{d}$. By combining ideas from {Harmonic Analysis} (partitions of unity and approximation theory) with \emph{Hashing-Based-Estimators}~[Charikar, Siminelakis FOCS'17], we provide a general framework for designing such data structures through hashing that reaches far beyond what previous techniques allowed.
 
A key design principle is a collection of $T\geq 1$ hashing schemes with collision probabilities $p_{1},\ldots, p_{T}$ such that $\sup_{t\in [T]}\{p_{t}(x,y)\} = \Theta(\sqrt{w(x,y)})$.
This leads to a data-structure that approximates $Z_{w}(y)$ using a  sub-linear number of samples from each hash family. 
 Using this new framework along with  \emph{Distance Sensitive Hashing} [Aumuller, Christiani, Pagh, Silvestri PODS'18],  we show that such a collection can be constructed and evaluated efficiently for any {log-convex} function $w(x,y)=e^{\phi(\langle x,y\rangle)}$ of the {inner product  on the unit sphere} $x,y\in \mathcal{S}^{d-1}$.  

Our method leads to  data structures with  {sub-linear query time}  that significantly improve upon  random sampling and can be used for {Kernel Density} or {Partition Function Estimation}.  We provide extensions of our result from the sphere to $\mathbb{R}^{d}$  and from scalar functions to vector functions.  
\end{abstract}
\newpage
\thispagestyle{empty}
\tableofcontents

\newpage
\setcounter{page}{1}

\section{Introduction}
The analysis of massive datasets very often involves summing simple functions over a very large number of objects~\cite{muthukrishnan2005data,woodruff2014sketching,mcgregor2014graph}. While in all cases one can compute the sum of interest exactly in time and space polynomial  or even linear  in the size of the input, practical considerations, such as space usage and update/query time, require developing significantly more efficient algorithms that can {provably approximate} the quantity in question arbitrarily well.  For $\alpha\geq 1$, we say that $\hat{\mu}$ is an $\alpha$-approximation to $\mu$  if $\alpha^{-1}\mu \leq \hat{\mu}\leq \alpha \mu$ and an $(1\pm \epsilon)$-approximation if $(1-\epsilon)\mu \leq \hat{\mu}\leq (1+\epsilon)\mu$.


Modern applications in Machine Learning pose an additional challenge in that such functions depend on a parameter vector $y\in \R^{d}$ that is unknown a priori or changes with time. Such examples include {outlier detection}~\cite{zou2014unsupervised}, {mode estimation}~\cite{cheng1995mean,arias2015estimation}, and {empirical risk minimization} (ERM)~\cite{vapnik2006estimation,shalev2014understanding}. Moreover, very often in order to train  faster and obtain better models~\cite{hardt2016train} it is required to estimate sums of vector functions (e.g gradients in ERM). Motivated by such applications, we seek sub-linear time algorithms for {summing pairwise functions} in high dimensions.

Given a set of points $X=\{x_{1},\ldots, x_{n}\}\subset \R^{d}$, a   non-negative function $w:\R^{d}\times \R^{d}\to [0,1]$,\footnote{For bounded non-negative functions we can always make this assumption, since if $w_{\max}=\sup\{w(x,y)|x,y\in \R^{d}\}\notin \{0,1\}$ we can assume without loss of generality that we are given $w/w_{\max}$. If $w_{\max}=0$ then the sum is identically zero.} and  a parameter $\epsilon>0$, we study the problem of designing a {data structure} that for any query $y\in \R^{d}$ in {sub-linear time}  provides a $(1\pm \epsilon)$-approximation to the sum:
\vspace*{-0.1in}
\begin{equation}
Z_{w}(y) = \frac{1}{n}\sum_{i=1}^{n}w(x_{i},y)
\end{equation} 
The actual value of the sum $Z_{w}(y)\in [0,1]$ for a given query $y$, will be denoted by $\mu$ and, as we see next, we can use  a lower bound   $\tau\leq \mu$ to bound the complexity of the problem.  

A prominent method to approximate such sums is constructing \emph{unbiased estimators of low variance}.
The simplest and extremely general approach to get such estimators is through {uniform random sampling}. Letting $\chi\in (0,1)$ be an upper bound on the failure probability, a second moment argument shows that  storing and querying  a uniform random sample of size $O\left(\frac{1}{\epsilon^{2}}\frac{1}{\tau}\log(1/\chi)\right)$  is sufficient   and necessary  in general~\cite{li2001improved,charikar2017hashing}, to approximate the sum $\mu=Z_{w}(y)$ for any $\mu\geq \tau$.  The dependence on $\epsilon, \chi$ is standard and easily shown to be necessary, so the question is \emph{for which class of functions can we improve the dependence on $\tau$?}

In this paper, we focus on the class of {log-convex} functions of the inner product between two vectors on the unit sphere.  Such functions can be written as $w(x,y)=e^{\phi(\langle x, y\rangle)}$ for some convex function $\phi:[-1,1]\to\R$ of the inner product between $x,y\in \mathcal{S}^{d-1}$.  Approximate summation of such functions has several fundamental applications in Machine Learning, including:
\begin{itemize}
\item \emph{Partition Function Estimation}~\cite{wainwright2008graphical,koller2009probabilistic}: a basic workhorse in statistics are exponential families where, given a parameter vector $y\in \R^{d}$, for all $x\in X\subseteq\R^{d}$ a probability distribution is defined by setting $p_{y}(x)\propto e^{\langle x, y\rangle}$. The normalizing constant $Z(y)=\sum_{x\in X}e^{\langle x,y\rangle}$ is called the \emph{partition function}. Approximating this quantity is important for hypothesis testing and inference.

\item \emph{Kernel Density Estimation:} a non-parametric way~\cite{devroye2012combinatorial} to estimate the ``density of a set $X$ at $y$" is  through  $Z(y) =\frac{1}{n\sigma^{d}}\sum_{i=1}^{n}\exp(-\frac{\|x_{i}-y\|^{2}}{\sigma^{2}})$. Such an estimate is used in algorithms for outlier detection~\cite{schubert2014generalized,gan2017scalable}, topological data analysis~\cite{joshi2011comparing} and clustering~\cite{arias2015estimation}.

\item \emph{Logistic activation and Stochastic Gradients:} let $\phi(\rho)=-\log(1+e^{- \rho})$ be the logistic function. For $X\subset \mathcal{S}^{d-1}$ we can express the sum of the output of $n$ neurons with weight vectors $x_{1},\ldots,x_{n}$ and input $y$ as
$
Z(y) = \frac{1}{n}\sum_{i=1}^{n}e^{-\log(1+e^{-  \langle x_{i},y\rangle})}= \frac{1}{n}\sum_{i=1}^{n}\frac{1}{1+e^{- \langle x_{i},y\rangle}}
$. This quantity can also be viewed as the sum of the gradient norms $\sum_{x\in X}\|\nabla_{y}\log(1+e^{\langle x,y\rangle})\|$, that is related to computing a stochastic approximation to the gradient at $y$ in Logistic Regression.
\end{itemize}

\begin{table}[t]
\caption{Examples of log-convex functions of  inner product  $\rho=\langle \frac{x}{\|x\|},\frac{y}{\|y\|}\rangle$ for $x,y\in r\mathcal{S}^{d-1}$.}
\label{tbl:examples}
\vskip 0.15in
\begin{center}
\begin{small}
\begin{sc}
\begin{tabular}{lccr}
\toprule
$w(x,y)$ & $\phi(\rho)$ & $L(\phi)$ \\
\midrule
$e^{\langle x, y\rangle}$& $r^{2}\rho$ & $r^{2}$\\
$e^{-\|x-y\|_{2}^{2}}$    & $2r^{2}(\rho-1)$ & $2r^{2}$ \\
$(\|x-y\|_{2}^{2}+1)^{-1}$ & $-\log(1+(1-\rho)2r^{2})$ & $2r^{2}$\\
$(1+\exp(-\langle x,y\rangle))^{-1}$    & $- \log(1+e^{-r^{2}\rho})$ & $r^{2}$\\
$(\langle x, y\rangle+cr^{2})^{-k}$& $-k\log(r^{2}(\rho+c))$    &$\frac{k}{c-1}$\\
\bottomrule
\end{tabular}
\end{sc}
\end{small}
\end{center}
\vskip -0.3in
\end{table}
\vspace*{-0.05in}
More examples of log-convex functions are presented in Table \ref{tbl:examples}. Obtaining fast algorithms for approximating   summations  gives speedups to all of the above applications. For such functions we denote $Z_{w}(y)$ as $Z_{\phi}(y)$. Let $L(\phi)$ be the lipschitz constant of the function $\phi$, we have that $Z_{\phi}(y)\geq  e^{-2L(\phi)}$ and hence random sampling requires $O(\frac{1}{\epsilon^{2}}e^{2L(\phi)})$ samples. For $L\geq \frac{1}{2}\log n$ random sampling offers \emph{no improvement} over the trivial algorithm.   In this work we design the first sub-linear algorithms for the problem of summing general log-convex functions on the unit sphere.

\vspace*{-0.15in}
\subsection{Our results}
At a high level,  we significantly generalize the recent approach of Hashing-Based-Estimators~\cite{charikar2017hashing} to handle more general functions. This is done by combining classical ideas from \emph{Harmonic analysis} (partitions of unity and approximation theory) with  recent results for \emph{similarity search}. We give a general technique for approximating pairwise summations that gives the following result for log-convex functions:

\begin{theorem}[Main Result]\label{thm:simple}
Given a log-convex function $\phi:[-1,1]\to \R$ with lipschitz constant $L(\phi) < (1-\delta) \log n$ for $\delta>0$, there exists a data structure that for $\epsilon>0$ and any set of $n$ vectors $X\subset \mathcal{S}^{d-1}$   can provide a $(1\pm\epsilon)$-approximation to $Z_{\phi}(y)$ for any query $y\in \mathcal{S}^{d-1}$  with constant probability  and query time $ n^{1-\delta+o(1)}/\epsilon^{2}$ using  space/pre-processing time $ n^{2-\delta +o(1)}/\epsilon^{2} $.
\end{theorem}

We show that under popular conjectures a {restriction on $L(\phi)$ is necessary} in order to obtain sublinear algorithms for the problem even on average over $n$ queries. In fact, it turns out that the correct asymptotics  is precisely  $L=O(\log n)$ 
even if one allows for polynomially large approximation factors. The proof and definition of the conjectures can be found in Section \ref{sec:lower}.
\begin{theorem}\label{thm:lower}
 Unless SETH and OVC fails, for every $\delta>0$ and $\alpha \geq 1$ there exists a constant $C(\delta, \alpha)>0$ such that for two sets $X,Y\subset \mathcal{S}^{d-1}$ of size $n$   with  $d= O_{\delta}(\log n)$ and $L > C(\delta,\alpha)  \cdot \log n$, there exists no $n^{2-O(\delta)}$ algorithm that  produces an $\alpha$-approximation to  $\frac{1}{n}\sum\limits_{y\in Y}\left(\frac{1}{n}\sum\limits_{x\in X} e^{L \cdot \langle x,y\rangle}\right)$.
\end{theorem}
The precise dependence is  $C(\delta,\alpha ) = O(e^{e^{\frac{\delta}{c(\delta)}}})(1+  \log \alpha /2 \log n)$ where $c(\delta)$ is a constant. Even if we allow for approximation factor $\alpha = n^{s}$ with $s>0$, we see that $C(\delta, n^{s})$ is still a constant.  The intuition behind this result is that when $L=\Omega(\log n)$ the function $e^{L\langle x, y\rangle}$ varies fast enough so that  the presence or absence of a {single pair of ``relatively close" points can dominate the sum}.  Below, we give concrete examples for which our data structure  has $n^{0.5+o(1)}$ query time, i.e. $L\leq \log(n) / 2$.

\begin{corollary}\label{cor:examples}
 Let $\Phi_{r,k,c}$ be the set of functions in Table \ref{tbl:examples} with parameters $r\leq \frac{1}{2}\sqrt{\log n}$ and $0\leq k\leq \frac{c-1}{2}\log n$. Then for any $\phi\in \Phi_{r,k,c}$ and $X\subset r \mathcal{S}^{d-1}$, there exists a data structure using space $n^{1.5+o(1)}/\epsilon^{2}$ that for any $y\in r\mathcal{S}^{d-1}$ can produce a $(1\pm\epsilon)$-approximation to  $Z_{\phi}(y)$  in time $n^{0.5+o(1)}/\epsilon^{2}$.
\end{corollary}
This corollary highlights the main point of our paper: we provide a  \emph{general technique}  that enables the design of data structures that {solve a variety of pairwise integration problems}. For the special case of the Gaussian kernel for points on a sphere, our data structure has the same dependence in $\epsilon, r$ (up to poly-logarithmic factors in $n$) as the currently best known algorithm~\cite{charikar2017hashing}.

\paragraph{Extensions} Our result is extended in a few different ways to be more broadly applicable
\begin{enumerate}
\item \emph{General subsets of  $\R^{d}$:} our method can be extended to bounded subsets of $\R^{d}$. Assuming that for all $x\in X$ we have $0<r_{0}\leq \|x\|\leq r_{X}$ and that $r_{0}\leq \|y\|\leq r_{Y}$,   in Section \ref{sec:reduction} we show that, by partitioning points in exponentially increasing spherical annuli (as in \cite{andoni2015optimal}) and by applying our result appropriately for each spherical annulus, we get a data structure with a space/query time overhead of a $O\left((\log(\max\{r_{X},r_{Y}\}/r_{0})L(\phi)r_{X}r_{Y})^{2}\right)$ factor   and where the Lipschitz constant increases  at most by an $r_{X}r_{Y}$ factor.
\item \emph{More general functions:} the previous  technique shows that our method applies also to the following wider family of functions $w(x,y)=p_{0}(\|x\|)e^{\phi(\langle x,y\rangle)+\mathcal{A}(y)}$, where $\log(p_{0}(\|x\|))$ is Lipschitz in each annulus and $\mathcal{A}(y)$ is arbitrary. Examples of such functions are the Gaussian kernel $e^{-\|x\|^{2}+2\langle x,y\rangle-\|y\|^{2}}$ and  the norm $\|\nabla_{y}\log(1+e^{\langle x,y\rangle})\|$ of the derivative of the logistic log-likelihood. The same principle can be applied to  solve weighted versions of the problem.
\item \emph{Vector functions:} Hashing-based-Estimators belong to a more general class of  randomly weighted importance sampling schemes, for which we show (Section \ref{sec:gradients}) that one can construct unbiased estimators for the sum of vector functions $\sum_{x}\vec{f}(x)$ with variance at most that of estimating  the sum of the norms $\sum_{x}\|\vec{f}(x)\|$. 
\end{enumerate}

\vspace*{-0.15in}

 \subsection{Motivation:  Partitions of Unity}\label{sec:partitions}
\begin{figure}[t]
\centering
\includegraphics[scale=0.5]{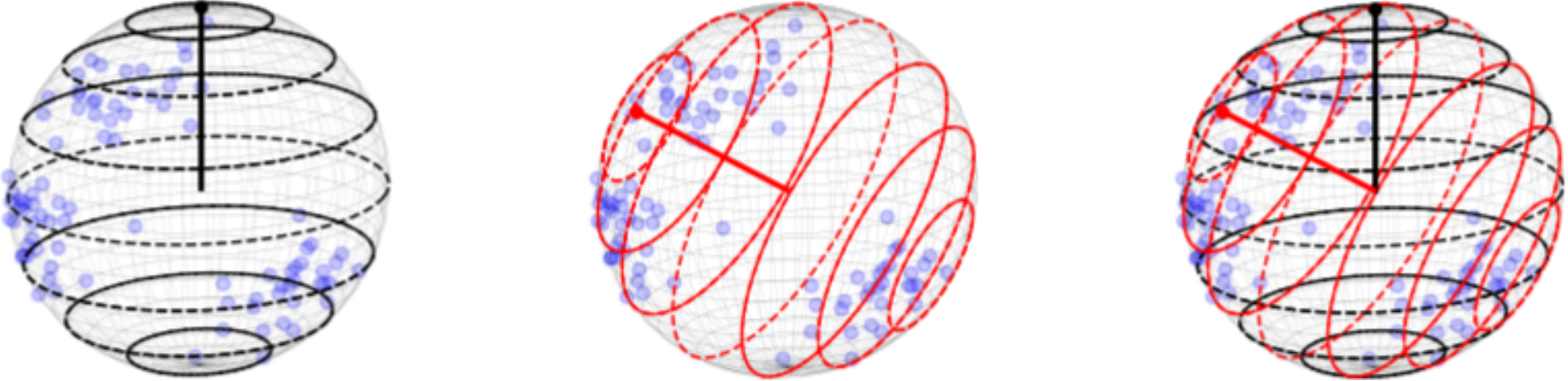}
\caption{Spherical partitions for two different query points (black and red) for a fixed dataset.}\label{fig:partitions}
\end{figure}
A general way to estimate sums over $X$ is to define a query-dependent partition  $\mathcal{P}(y)=\{P_{1}(y),\ldots, P_{T}(y)\}$ of $X$ and express the sum as $\sum_{t\in [T]}\left(\sum_{x\in P_{t}}w(x,y)\right)$. If for the specific partition there exist  $M\geq 1$ such that:
\vspace*{-0.1in}
\begin{equation}\label{eq:}
\frac{1}{M}\cdot w(x_{2},y) \leq w(x_{1},y) \leq M \cdot w(x_{2},y),\qquad \forall t\in [T], \forall x_{1},x_{2}\in P_{t}(y)  
\end{equation}
taking $O(M/\epsilon^{2})$ random samples would give us an accurate estimate of each term $\sum_{x\in P_{t}}w(x,y)$ and using at most $O(MT/\epsilon^{2})$ samples we would obtain a good estimate of the sum. The problem is that  generating and sampling  from such a partition efficiently for any query $y$ can be  computationally challenging. For example if $w(x,y)=e^{-\|x-y\|^{2}}$ and points $X\subset  r \mathcal{S}^{d-1}$ lie on a sphere, then such partitions are equivalent to being able to sample from a certain \emph{spherical range} around the query $y\in r\mathcal{S}^{d-1}$ (Figure \ref{fig:partitions}).  Computing such partitions in high dimensions can be expensive~\cite{ahle2017spherical,aumuller2018distance}.

\paragraph{Partitions of unity} 
Instead of a partition $\mathcal{P}$, consider a collection of functions $\tilde{w}_{t}(x,y)$ such that $\sum_{t\in [T]}\tilde{w}_{t}(x,y)=1,  \forall x\in X$.
Each such function concentrates its mass on a small portion of the space -- this can be thought of as a soft partition.
Such a collection of functions is called a \emph{partition of unity}  (Figure \ref{fig:riemann}) and is widely used in Harmonic analysis. 
We will use partitions of unity to define estimators for which we can control their first and second moments through {linearity of expectation} and provide a generic recipe to use them within the framework of {Hashing-based-Estimators} to bound the overall variance. 
\begin{figure}[t]
\centering
\includegraphics[scale=0.7]{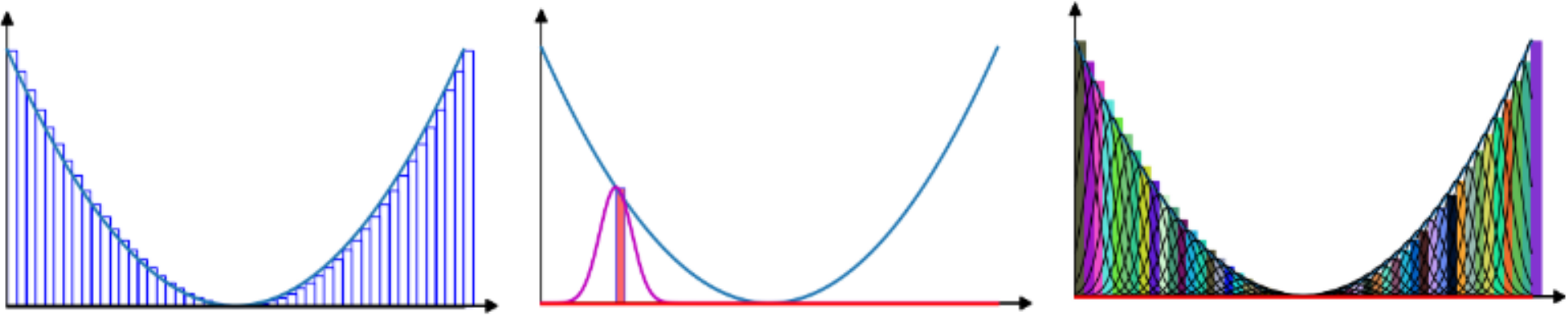}
\caption{Partitions of unity as a tool of rewriting integrals in terms of of localized functions}\label{fig:riemann}
\end{figure}

\vspace*{-0.15in}
\subsection{Our techniques}
The main conceptual contribution of this work is a new framework for approximating pairwise summations.  Our framework is based on a class  of estimators that we introduce, called \emph{Multi-resolution} Hashing-Based-Estimators, that significantly generalizes  previous work~\cite{charikar2017hashing}. 
The main idea is that, instead of a single hashing scheme, we have a collection of hash families $\mathcal{H}_t$ for $t \in [T]$, where each $\mathcal{H}_t$ is responsible for a different portion of the angular range around the query; $\mathcal{H}_t$ has relatively high collision probability within the range assigned to it and relatively low outside.
We divide up the task of estimating the summation of interest amongst these various hash families by assigning data points $x \in X$ to $t \in [T]$ via a soft partition (i.e. a partition of unity).
Our end goal is to produce an unbiased estimator and bound its variance by selecting the hashing scheme and partition of unity appropriately. 
While this overall scheme sounds complicated, we show that a particular choice of weights for the soft partition (as a function of collision probabilities) makes the analysis modular and tractable: \emph{for the purpose of analysis, the collection of hash families behaves like a single hash family whose collision probability is the supremum of the collision probabilities for $\mathcal{H}_t, t \in [T]$.}
We now flesh out this informal description.
 
\paragraph{Multi-resolution Hashing-Based-Estimators (MR-HBE)} Given a collection $\mathcal{H}_{1},\ldots, \mathcal{H}_{T}$ of hashing schemes with collision probabilities $p_{1},\ldots,p_{T}:\R^{d}\times \R^{d}\to [0,1]$ and  functions $\tilde{w}_{t}:\R^{d}\times \R^{d}\to \R_{+}$ for $t\in [T]$, such that $\sum_{t\in T}\tilde{w}_{t}(x,y)=1$ (\emph{partition of unity}) and $w_{t}(x,y):=\tilde{w}_{t}(x,y)w(x,y)>0 \Rightarrow   p_{t}(x,y)>0$,  we form an unbiased estimator by:
\begin{itemize}
\item \emph{Preprocessing:} for all $t\in [T]$,  sample a hash function $h_{t}\sim \mathcal{H}_{t}$ and evaluate it  on $X$ creating hash table $H_{t}$ . Let $H_{t}(z)\subseteq X$ denote the hash bucket where $z\in \R^{d}$ maps to under $h_{t}$.
\item \emph{Querying:} given a query $y\in \R^{d}$, for all $t\in [T]$ let $X_{t}\sim H_{t}(y)$ be a random element from $H_{t}(y)$ or $\bot$ if $H_{t}(y)=\emptyset$. Return $Z_{T}(y)=\frac{1}{|X|}\sum_{t\in T} \frac{w_{t}(X_{t},y)}{p_{t}(X_{t},y)}|H_{t}(y)|$.
\end{itemize}
\vspace*{-0.1in}
where it is understood that if   $X_{t}=\bot$ the corresponding term is $0$. The conditions on $\{\tilde{w}_{t}\}$ and $\{p_{t}\}$ ensure that the estimator is unbiased.  The motivation behind these estimators is to use the extra freedom in selecting $\{\tilde{w}_{t}\}$ and $\{p_{t}\}$ so that we can obtain better bounds on the overall variance. This is quite challenging as the variance of each of the $T$ terms in the sum depends on the whole data set through $|H_{t}(y)|$. This raises the question:
\vspace*{-0.05in}
\begin{center}
\emph{Do there exist design principles for $\{\tilde{w}_{t}\}$ and $\{p_{t}\}$ that lead to low variance?}
\end{center}
Through our analysis we introduce two key design principles:
\paragraph{Variance bounds and $p^{2}$-Weighting} For a fixed collection of weight functions $\{\tilde{w}_{t}\}$ and collision probabilities $\{p_{t}\}$, by utilizing a lemma from ~\cite{charikar2017hashing}, we get an explicit bound on the variance of the estimator for a query $y\in \R^{d}$ only as a function of $\{w_{t}(\cdot, y)\},\{p_{t}(\cdot, y)\}$ and $\mu:=Z_{\phi}(y)$. We then minimize a separable relaxation of our upper bound to obtain the $p^{2}$\emph{-weighting scheme} where
\vspace*{-0.05in}
\begin{equation}\label{eq:pou}
\tilde{w}_{t}(x,y)=\frac{p^{2}_{t}(x,y)}{\sum_{t'=1}^{T}p_{t'}^{2}(x,y)} \qquad \text{for all} \qquad x,y\in \R^{d}
\end{equation} 

\vspace*{-0.1in}
\paragraph{Approximation by a supremum of functions}  Using the $p^{2}$-weighting scheme and after some algebraic manipulations, we are able to get an upper bound on the variance that depends only on $w(x,y)$, $\mu=Z_{w}(y)$ and on the \emph{pointwise supremum} of the collision probabilities $p_{*}(x,y):=\sup_{t\in [T]}\{p_{t}(x,y)\}$. An interesting fact that comes out from the analysis is that the resulting bound is closely related to the variance of  a single HBE, i.e. $T=1$, with collision probability equal to $p_{*}(x,y)$. Exploiting  this connection and by providing a simplified proof for a theorem of ~\cite{charikar2017hashing}  that bounds the variance of \emph{scale-free HBE}, we identify the second design principle, namely designing $\{p_{t}\}$ such that:
\begin{equation}\label{eq:supremum}
p_{*}(x,y)=\sup_{t\in T}\{p_{t}(x,y)\} = \Theta (\sqrt{w(x,y)})
\end{equation}
Observe that so far our discussion has been about the variance, or on how many independent realizations of Multi-resolution HBE we need to efficiently estimate $Z_{\phi}(y)$, and we have not mentioned the time needed to compute each one. The natural question is then:  for which family of functions $w(x,y)$, does there exist a  \emph{family of hashing schemes}  $\{(\mathcal{H}_{t}, p_{t})\}$ satisfying \eqref{eq:supremum}  that can be \emph{efficiently constructed and evaluated?}
\begin{center} 
\end{center}
 
\begin{figure}[t]
\caption{Approximation of \emph{logistic} $\phi_{1}(\rho)=-\log(1+e^{-\rho})$ and \emph{squared inner product} $\phi_{2}(\rho)=\rho^{2}$ functions by elements of \eqref{eq:idealized}. The functions are normalized to be less than $0$ and at least $-1$.}\label{fig:examples}
\includegraphics[scale=0.5]{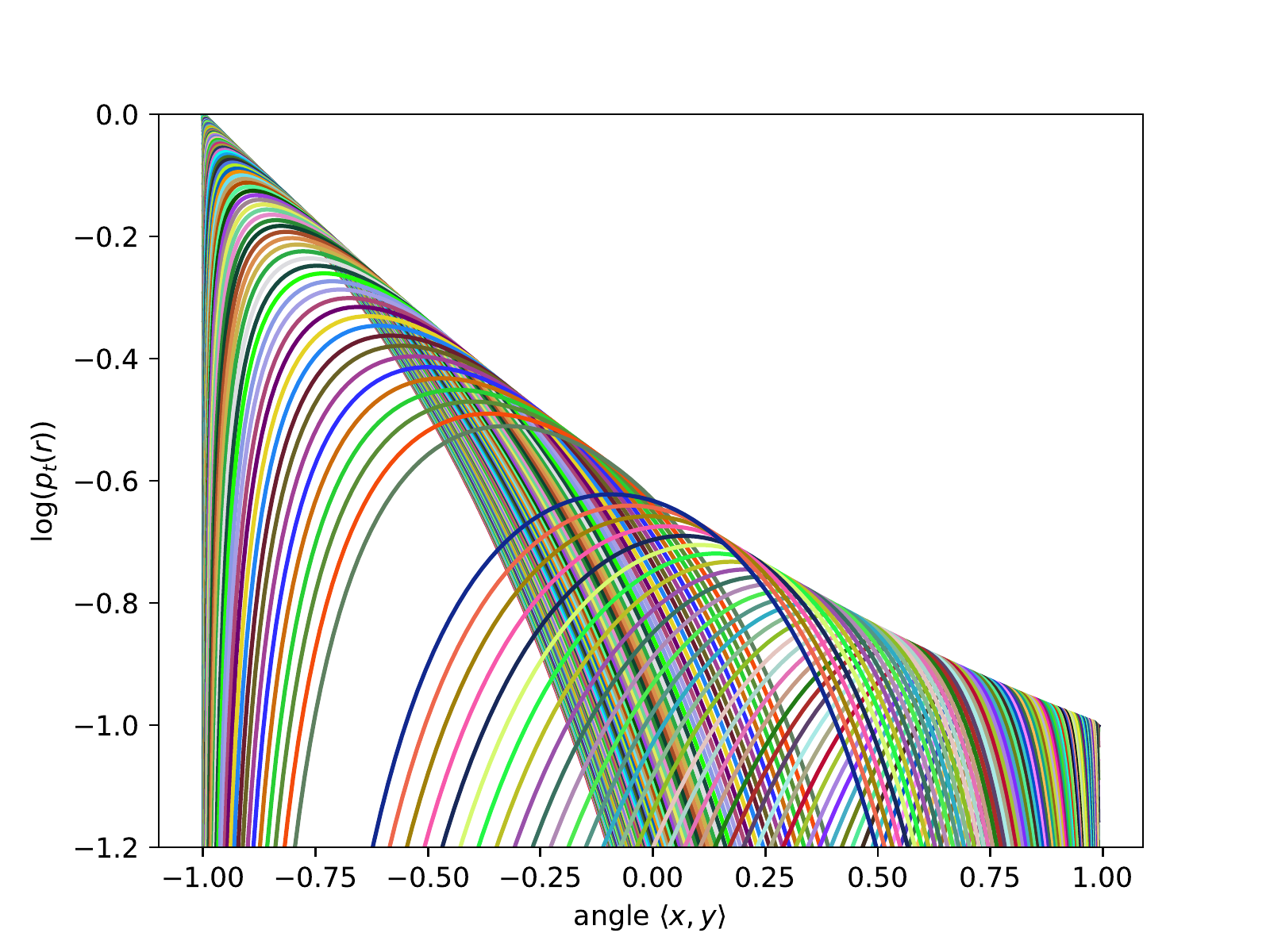}
\includegraphics[scale=0.5]{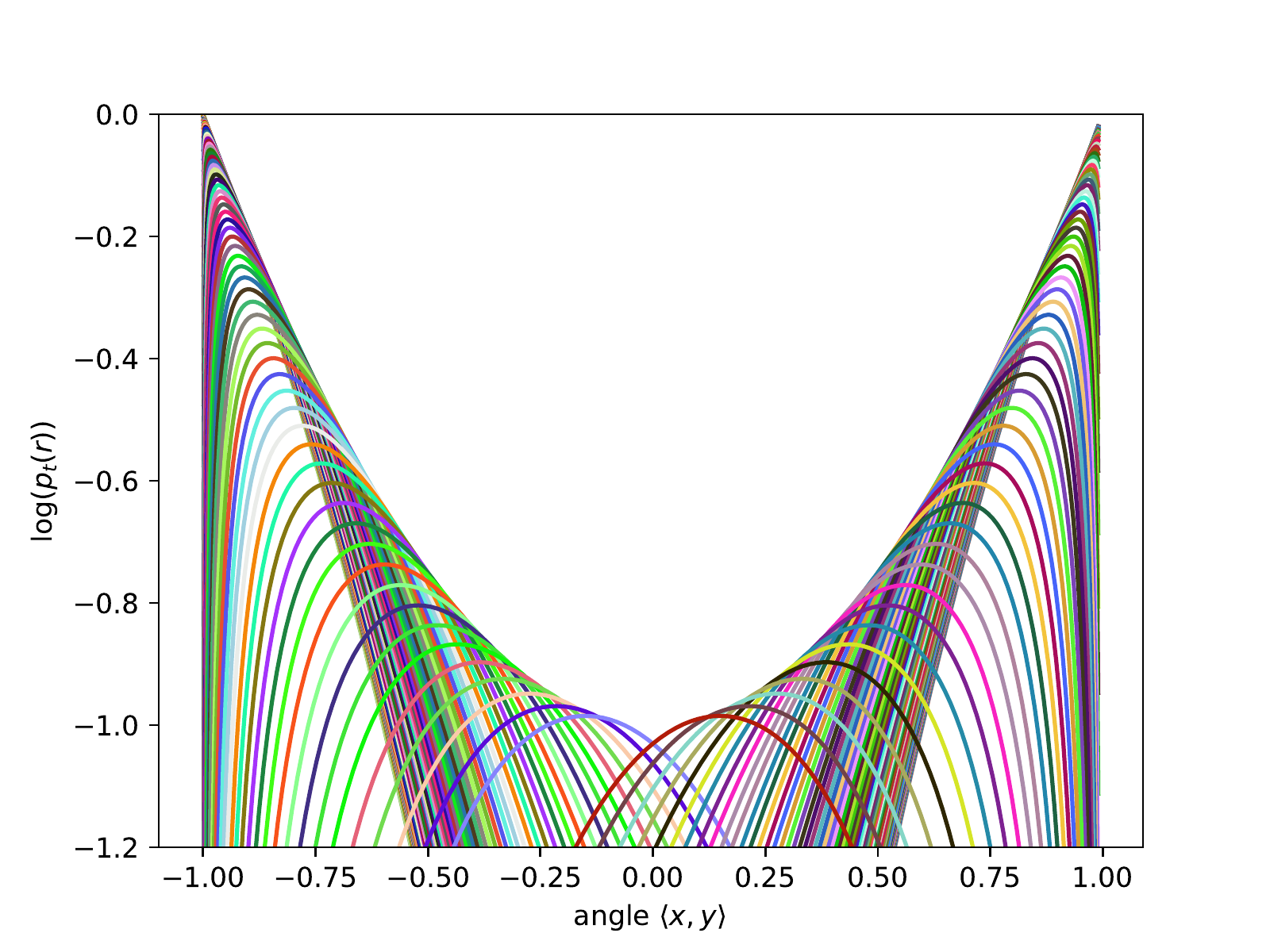}
\end{figure}

\vspace*{-0.3in}
\paragraph{Approximating Log-convex Functions via Distance Sensitive Hashing} We show that this is indeed possible for {log-convex functions of the inner product} by utilizing a family of hashing schemes introduced recently by Aumuller et al.~\cite{aumuller2018distance}, referred to as \emph{Distance Sensitive Hashing} (DSH). This family is defined through two parameters $\gamma\geq 0$ and $s >0$, with collision probability $p_{\gamma,s}(\rho)$ having the following dependence on the inner product $\rho=\langle x,y\rangle$ between two vectors $x,y\in\mathcal{S}^{d-1}$
\begin{equation}\label{eq:idealized} 
\log\left(1/p_{\gamma,s}(\rho)\right) =  \Theta\left(\left(\frac{1-\rho}{1+\rho}+\gamma^{2}\frac{1+\rho}{1-\rho} \right)\frac{s^{2}}{2} \right)
\end{equation}
We provide a slightly modified analysis of the collision probability that gives us better control of the hidden constants in the above equation. This allows us to show that for any convex function $\phi$, we can use a small number of DSH families to achieve \eqref{eq:supremum} that in this case is equivalent  to
\begin{equation}\label{eq:log-approx}
\bigl|\log \sup_{t\in [T]}\{p_{t}(\rho)\} -\frac{1}{2}(\phi(\rho)-\phi_{\max})\bigr| = O(1)
\end{equation}
The approximation is achieved by: (a) producing a sequence of explicit ``interpolation points" $\rho_{1},\ldots,\rho_{T}\in [-1,1]$, (b) using a single scheme to approximate the function $\phi$ locally  (value and derivative) around each $\rho_{t}$ \emph{(multi-resolution)}, (c) and then  using convexity of $\phi$ and ``concavity" of $p_{\gamma,s}$ to bound the error in \eqref{eq:log-approx} (Section \ref{sec:family}). 
The \emph{number of estimators} as well as the \emph{approximation error} in \eqref{eq:log-approx} are \emph{sub-linear} in the Lipschitz constant $L(\phi)$ of the function. This dependence of the error is the result of achieving a trade-off between evaluation time of the hash functions and fidelity of approximation, that affects the variance (Section \ref{sec:scale-free}).
An interesting fact is that to be able to achieve the above approximation guarantee using DSH, \emph{convexity of the function $\phi$ is instrumental} (Lemma \ref{lem:positive} and Proposition \ref{prop:infimum}). We give two examples of the resulting approximation in Figure \ref{fig:examples}.

\paragraph{Summary}
Our work provides a general technique that reduces the computational task of summing a pairwise function over a large dataset to the  task of constructing a family of hash functions whose square of the \emph{pointwise suprememum of collision probabilities approximates the function in question}. 
\subsection{Previous work}

Recent approaches on obtaining sub-linear algorithms for pairwise summation are based on two different ideas:   \emph{Hashing-based Importance Sampling} and \emph{Well-conditioned Partitions} .

\subsubsection{Hashing-based Importance Sampling}
Importance Sampling aims to reduce the variance of uniform random sampling by sampling  points according to some biased  distribution that assigns \emph{greater probability} to points with \emph{higher value} $w(x,y)$.  The challenge in our setting is that such a distribution needs to be adaptive to the query $y\in \R^{d}$ and to admit an efficient sampling algorithm at query time. 

\paragraph{Hashing-Based-Estimators (HBE)} In a previous work of the authors~\cite{charikar2017hashing},  the general approach of using hashing to create importance sampling schemes with provable low-variance was introduced under the name of Hashing-Based-Estimators.  Given a \emph{single} hashing scheme $\mathcal{H}$
with collision probability $p(x,y)=\P_{h\sim \mathcal{H}}[h(x)=h(y)]$    an \emph{unbiased estimator} for $Z_{w}(y)$ is constructed  through a two-step sampling process (corresponds to the $T=1$ case of Multi-Resolution HBE). The main technical contributions of~\cite{charikar2017hashing} that the current paper builds on (see Section \ref{sec:prelims}) are:
\begin{enumerate}
\item A  reduction that shows that the estimation problem can be reduced to the problem of efficiently constructing a $V$-bounded estimator (Theorem \ref{thm:v-bounded}).
\item A variational bound on the variance of importance sampling schemes (Lemma \ref{lem:holder})
\item The  concept of the ``scale-free" property $p(x,y)=\Theta(\sqrt{w(x,y)})$  for a \emph{single} hashing scheme.
\end{enumerate}

\paragraph{Limitations of HBE} The approach of HBE   hinges upon constructing a \emph{single hashing scheme} that has the \emph{scale-free property} (defined above). This can be quite difficult to achieve with hash functions that can be efficiently stored and evaluated.  In fact, the authors were able to carry out this approach  \emph{ for exactly three  functions}: the Gaussian $e^{-\|x-y\|_{2}^{2}}$, Exponential $e^{-\|x-y\|_{2}}$, and Generalized $t$-Student $1 / (1+\|x-y\|_{2}^{p})$ kernels using \emph{Locality Sensitive Hashing} schemes of Andoni-Indyk~\cite{andoni2006near} and Datar et al.~\cite{datar2004locality}. This is due the fact that these LSH schemes exhibited collision probabilities  that matched the aforementioned functions. Hence, there are severe restrictions on the classes of functions for which sub-linear algorithms can be obtained through HBE.

\paragraph{Comparison} In this work, we essentially remove the main bottleneck of the Hashing-based approach and make it more broadly applicable. This is done by using the idea of Partitions of Unity via  Multi-Resolution HBE, and identifying key design principles \eqref{eq:pou} and \eqref{eq:supremum} that provably lead to an overall low-variance estimator. In doing so we also provide a more general theorem for the variance of scale-free estimators (Theorem \ref{thm:scale-free}).

The bulk of our technical work goes into showing that this approach is generic enough to capture a large class of functions, namely \emph{log-convex functions of the inner product}.   This is critical as for applications in Optimization or Machine Learning, one cares about functions that depend on the \emph{inner products} between two vectors rather than their distance. We achieve this  by developing an Approximation Theory of Convex Functions using the family of collision probabilities given by Distance Sensitive Hashing schemes.

\subsubsection{Partition-based approaches and Smoothness}
The idea of partition-based approaches, is to efficiently partition points in a small number of parts such that some simple primitive (Random Sampling or Polynomial approximation) can be used to accurately estimate the contribution of each part. This approach in low dimensions, is known under the names of Fast-Multipole Methods~\cite{greengard1987fast} or Well Separated Pair Decomposition~\cite{callahan1995decomposition} and the   complexity scales typically as $\log(1/\epsilon)^{O(d)}$~\cite{greengard1991fast} for additive error $\epsilon$. 

Due to the explosion in Machine learning applications the problem was revisited in the high-dimensional case  through works on ``Dual-tree Algorithms"~\cite{gray2003nonparametric,yang2003improved,lee2006dual} that aimed to exploit an underlying low dimensional structure~\cite{ram2009linear} (when it exists). However, no theoretical results were known for the general case.

\paragraph{``Non-smooth" functions}The lower bound presented here, inspired by~\cite{backurs2017fine}, shows that this is for good reason. In high dimensions $d=\Omega(\log n)$, even for simple functions (e.g. Gaussian kernel),  and  under no restrictions on the rate that the function changes  we do not expect to be able to get sub-linear algorithms barring major progress in complexity theory (e.g. refuting SETH).

\paragraph{"Smooth" functions}In a recent work~\cite{backurs2018efficient}, it was established that indeed in high dimensions quick variation of the pairwise function  is the only obstacle in obtaining efficient algorithms.  In particular, the authors of~\cite{backurs2018efficient} introduced the following notion of $(C,L)$-smoothness that captures functions that   vary polynomially fast with distance:
\[
\max\left\{\frac{w(x,y)}{w(x',y)},\frac{w(x',y)}{w(x,y)}\right\}\leq C \left\{\frac{\|x-y\|}{\|x'-y\|},\frac{\|x'-y\|}{\|x-y\|}\right\}^{L}
\]
and showed that one can get $\mathrm{poly}(2^{L}, \log n, \frac{1}{\epsilon})$ algorithms   giving exponential improvement over the linear time algorithm for small values of $L=o(\log n)$. This was achieved by showing that one can efficiently construct query-dependent partitions (in time roughly $2^{O(L)}$)  that are \emph{``good on average"} when random sampling is used to approximate the contribution of each part. Interestingly, ideas related to hashing were instrumental to  both constructing and analyzing the partitions. The authors also  provided an intimate connection to the problem of Approximate Nearest Neighbor Search (ANNS) by   showing that for ``radial" and smooth functions one can   solve the problem given oracle access to an $c$-ANNS data structure using $\mathrm{poly}(c^{L}, \log n, \frac{1}{\epsilon})$ calls.

\paragraph{Comparison} The class of log-convex functions studied in this paper   \emph{does not} satisfy  this definition of smoothness (exponential vs polynomial). Moreover, even for $(O(1), \frac{1}{2}\log n)$-smooth functions this approach does not give any improvement over the linear time algorithm (cf. Corollary \ref{cor:examples}).

\subsection{Related work}
\subsubsection{Partition Function Estimation}
 For the special case of \emph{log-linear models}, there is a different approach that relies on LSH to approximate the partition function~\cite{mussmann2016learning,mussmann2017fast}. In the heart of this approach are two reductions. For $\alpha\geq 1$, the first one is reducing  the problem of obtaining a $\alpha$-approximation to the inverse of the Parition Function to obtaining an $\log(\alpha)$-\emph{additive approximation} for the problem of Maximum Inner Product Search (Gumbel trick). The second one, is reducing the problem of MIPS to the problem of $(1+\gamma(\alpha))$-approximate nearest neighbor search (ANNS). Using the best known data-structure for ANNS~\cite{andoni2015optimal}, this method requires  \emph{worst case} time/space $\Omega(n^{1-O(\gamma(\epsilon))})$, which is tight~\cite{andoni2017optimal}.  For vectors in $r\mathcal{S}^{d-1}$, the dependence is $\gamma(\alpha)=O(\frac{\log \alpha}{r^{2}})$. Hence, at least for adversarial data-sets this approach cannot bring forth significant improvements unless $r=O(\sqrt{\log \alpha})$. Nevertheless,  the authors \cite{mussmann2017fast} have shown experimentally that their method is still competitive compared to uniform sampling.

\vspace*{-0.1in}

\subsubsection{Core-sets}
A central notion in computational geometry~\cite{agarwal2005geometric}, learning theory~\cite{li2001improved} and approximation algorithms~\cite{har2011geometric} is that of a \emph{Coreset} or $\epsilon$-sample. Given a set $X\subseteq \mathcal{X}$ and a collection of functions $\mathcal{F}$ from a space $\mathcal{X}$ to $[0,1]$, an $\epsilon$-Coreset $(S,u)$ consists of a set $S\subset X$ and a function $u:S\to \R_{+}$ such that:
\vspace*{-0.05in}
\begin{equation}\label{eq:coreset}
(1-\epsilon) \sum_{x\in X} f(x) \leq \sum_{x\in S}u(x)f(x) \leq (1+\epsilon) \sum_{x\in X}f(x), \qquad \forall f\in \mathcal{F}
\end{equation}

\vspace*{-0.05in}
\paragraph{Kernel Density} In the context of Kernel Density Estimation~\cite{devroye2012combinatorial},  Coresets~\cite{joshi2011comparing,phillips2013varepsilon}  have received renewed attention in recent years resulting in near optimal constructions~\cite{phillips2018near} for certain cases. The literature has mostly been focused on obtaining \emph{additive error} $\epsilon>0$.  A general upper bound of $O(\frac{1}{\epsilon^{2}})$ was shown~\cite{bach2012equivalence,phillips2018improved} on coresets for characteristic kernels (e.g. Gaussian, Laplacian) using a greedy construction (\emph{kernel herding}~\cite{chen2010super}). The other approach~\cite{phillips2013varepsilon,phillips2018improved,phillips2018near} applies to Lipschitz kernels of bounded influence (decay fast enough), and constructs the core-set by starting with the full set of points  and \emph{reducing it by  half} each time. Using smoothness properties of the kernel one can then bound the error introduced by each such operation through the \emph{notion of discrepancy}~\cite{chazelle1996linear,chazelle2000discrepancy}. In this way, an upper bound of $O(\sqrt{d}/\epsilon\sqrt{\log(1/\epsilon)})$  was recently obtained~\cite{phillips2018near} for such  kernels (e.g. Gaussian, Laplacian). Furthermore,  corresponding lower bounds $\Omega(\sqrt{d}/\epsilon)$ and $\Omega(1/\epsilon^{2})$ were proved   for $d\leq 1/\epsilon^{2}$~\cite{phillips2018near} and  $d\geq 1/\epsilon^{2}$~\cite{phillips2018improved} respectively.
  
\vspace*{-0.05in}  
\paragraph{Machine Learning and Logistic Regression} The central paradigm of Machine Learning is that of Empirical Risk Minimization. Coresets provide a way to approximate the empirical risk in certain cases and speed up the training and evaluation of machine learning models~\cite{feldman2011scalable,baykal2018data}. The special case of logistic regression has  recently received special attention~\cite{huggins2016coresets,tolochinsky2018coresets,munteanu2018coresets}. A common theme of these approaches is that they show a \emph{lower bound on the sensitivity}~\cite{langberg2010universal} for logistic regression, and then  add restrictions on the norms~\cite{huggins2016coresets,tolochinsky2018coresets} or provide parametrized results depending on the complexity of the instance~\cite{munteanu2018coresets} . Our results are of similar nature as they are parametrized by the Lipshcitz constant of the convex function under consideration (corresponds to bound on norms) as well as by the complexity of the instance (density $\mu$, see Theorem \ref{thm:main})
 
\paragraph{Comparison with current approach}Our work gives better ``for any" guarantees  and (in some sense) sidesteps the issue of bounding the sensitivity, by allowing \emph{randomization and adaptivity} both within the weights  and on the set of points to be used for a given query. Our estimator can be written as $\sum_{x\in S_{y}} u(x) w(x,y)$ where both the set $S_{y}$ and weights $u(x)$ are \emph{random and depend} on $y$. In particular, in Section \ref{sec:gradients} we show how  both Hashing-based-Estimators and Importance sampling can be cast  under the same framework. Besides the difference in the  guarantees offered, our approach and Coresets are in some sense \emph{orthogonal}.   Even in the case where we obtain Coresets of small size,  our methods, when applicable,  can be used to \emph{accelerate the evaluation} of such \emph{Coresets}.

\vspace*{-0.1in}
\subsection{Outline of the paper}
In the next section, we describe the basis of our approach and introduce the main tools we need. In Section \ref{sec:variance}, we derive the key design principles for Multi-resolution HBE and show how they yield provable bounds on the variance. In Section \ref{sec:family}, we use an idealized version of the collision probabilities provided by Distance Sensitive Hashing  to approximate log-convex functions. In Section \ref{sec:scale-free}, we finish the construction of our estimators for the unit sphere and prove our main result. In Sections \ref{sec:reduction} and \ref{sec:gradients}, we show respectively how to extend this construction to Euclidean space and to estimate vector functions, whereas in Section \ref{sec:lower} we give the proof of the lower bound. Finally, in Section   \ref{sec:proofs}, we provide  the proofs for some intermediate lemmas and conclude with some open questions in Section \ref{sec:open}.

\vspace*{-0.1in}
\section{Preliminaries}\label{sec:prelims}
We introduce some parameters that capture the complexity of a function for our purposes. 
\begin{definition}\label{def:lipschitz}
Let $S\subset \R$, a function $f:S\to \R$ is called \emph{Lipschitz with constant} $0\leq L< \infty$ if for all  $x,y\in S$,   $|f(x)-f(y)|\leq L|x-y|$. For given $f$, we denote by $L(f)$ the  minimum such constant.
\end{definition}
Let also $R(f)=f_{\max}-f_{\min}$ denote the range of $f$.
\begin{proposition}\label{prop:affine_transform}
Given $a,b\in \R$, we have $L(a f + b) = |a|L(f)$ and $R(af+b) = |a|R(f)$.
\end{proposition}
\begin{proof}
If $a>0$, $R(af+b) = af_{\max}+b - (af_{\min}+b) = a R(f)$. If $a<0$, $R(af+b)= af_{\min}+b - (af_{\max}+b) = -aR(f)=|a|R(f)$. Finally, $|af(x)+b-(af(y)+b)| \leq |a||f(x)-f(y)|\leq |a|L|x-y|$.
\end{proof}
 Throughout the paper  for a query $y\in \R^{d}$ we use $\mu:=\mu(y)=Z_{w}(y)/w_{\max}$. For log-convex functions, we assume that $L(\phi)$ is greater than some small constant. Otherwise  $O(1/\epsilon^{2}\log(1/\chi))$ uniform random samples are sufficient to estimate any $\mu\in [e^{-R(\phi)},1]$.

\vspace*{-0.15in}
\subsection{Basis of the approach}

The  starting point of our work is the \emph{method of unbiased estimators}. Assume that we would like to estimate a quantity $\mu=\mu(y)$ using access to samples from a distribution $\mathcal{D}$, such that  for $\hat{Z}\sim \mathcal{D}$, $\E[\hat{Z}]=\mu$  and $\mathrm{Var}[\hat{Z}] \leq \mu^{2}V_{\mathcal{D}}(\mu)$. The quantity $V_{\mathcal{D}}(\mu)$ (depending possibly on $\mu$) bounds the \emph{relative variance} $\mathrm{RelVar}[\hat{Z}]:=\frac{\mathrm{Var}[\hat{Z}]}{(\E[\hat{Z}])^{2}}$. For $\epsilon>0$, we get through Chebyshev's inequality that the average of $O(\epsilon^{-2}V_{\mathcal{D}}(\mu))$ samples are sufficient to get $(1\pm\epsilon)$-multiplicative approximation to $\mu$ with constant probability. Moreover, using the median-of-means technique  ~\cite{alon1996space}, we can make the failure probability to be less than $\chi>0$ by only increasing the number of samples by a $O(\log(1/\chi))$ factor.

\paragraph{V-bounded Estimators}  The above discussion seems to suggest that as long as one has an unbiased estimator $\hat{Z}\sim \mathcal{D}$ for $\mu$ and a bound $V_{\mathcal{D}}(\mu)$ on the relative variance,  one can accurately estimate $\mu$. The caveat of course is that in cases where $V_{\mathcal{D}}$ is indeed a function of $\mu$, setting the requisite number of samples  requires knowledge of $\mu$. 
An unbiased estimator for which $\mu^{2}V_{\mathcal{D}}(\mu)$ is decreasing  and $V_{\mathcal{D}}(\mu)$ is increasing is called \emph{$V$-bounded}~\cite{charikar2017hashing}. An estimator has complexity $\mathcal{C} $, if using space  $O(\mathcal{C} n)$ we can evaluate it, i.e. sample from $\mathcal{D}$, in $O(\mathcal{C})$ time.  A general way to construct data-structures to solve estimation problems using $V$-bounded estimators was recently proposed.
\vspace*{-0.05in}
\begin{theorem}[\cite{charikar2017hashing}]\label{thm:v-bounded}
Given a $V$-bounded estimator
of complexity $\mathcal{C}$ and parameters $\epsilon,\tau,\chi \in (0,1)$, there exists a data structure that using space $O(\frac{1}{\epsilon^{2}}V_{\mathcal{D}}(\tau)\mathcal{C}\log(1/\chi)\cdot n)$  can provide a $(1\pm \epsilon)$ approximation to any $\mu\geq \tau$ in time $O(\frac{1}{\epsilon^{2}}\mathcal{C}V_{\mathcal{D}}(\mu)\log(1/\chi))$ with probability at least $1-\chi$. The data-structure can also detect when $\mu <\tau$.
\end{theorem}
\vspace*{-0.05in} 
Our goal 
is to construct such estimators through hashing and bound their complexity. 
The above theorem
turns our construction into an efficient data-structure for \emph{estimating pairwise summations}. 
\vspace*{-0.2cm}
\subsection{Analytical Tools}
The following variational inequality was first proved in \cite{charikar2017hashing} and bounds the \emph{maximum of a quadratic form} over the intersection of two weighted $\ell_{1}$-balls. This is going to be the key lemma that will allow us to obtain worst-case bounds on the variance of our estimators. 
\begin{lemma}[\cite{charikar2017hashing}]\label{lem:holder}
Given positive vector $w\in \R^{n}$, number $\mu>0$, define $f^{*}_{i}:=\min\{1,\frac{\mu}{w_{i}}\}$. For any matrix $A\in \R^{n\times n}$ :
\vspace*{-0.05in}
\[
\sup_{\|f\|_{w,1}\leq \mu, \|f\|_{1}\leq 1}\{f^{\top}Af\}\leq 4\sup_{ij\in [n]}\left\{f_{i}^{*}|A_{ij}|f_{j}^{*}\right\}
\] 
\end{lemma} 
\smallskip

The following crucial lemma, that \emph{upper bounds the value of a convex function away from the natural boundary}, lies in the core of our ability to use the family of functions \eqref{eq:idealized} to approximate convex functions of the inner product.
\vspace*{-0.2cm}
\begin{lemma} \label{lem:positive}
 Let $\phi:[-1,1]\to \R$ be a  non-constant, non-positive, convex, differentiable function, then 
\begin{equation}
 2\phi(\rho_{0}) <  -(1-\rho_{0}^{2})|\phi^{'}(\rho_{0})|, \ \forall \rho_{0}\in (-1,1).
 \end{equation}
\end{lemma}
\begin{proof}
Let $g(\rho) = \phi^{'}(\rho_{0})(\rho-\rho_{0})+\phi(\rho_{0})$ be the linear approximation of $\phi$ around $\rho_{0}\in(-1,1)$,   by convexity we have that $g(\rho)\leq \phi(\rho) \leq 0$. First let's assume that $g$ is increasing, then:
\begin{equation}
g(1)\leq 0 \Rightarrow \phi'(\rho_{0}) \leq -\frac{\phi(\rho_{0})}{1-\rho_{0}}\Rightarrow 2\phi(\rho_{0}) +(1-\rho_{0}^{2})|\phi^{'}(\rho_{0})| \leq  2\phi(\rho_{0})\left[1-\frac{1+\rho_{0}}{2} \right] <0\label{eq:der_increasing}
\end{equation}
where the last inequality follows from the fact that a non-constant convex function attains its maximum only at the boundary of a convex domain. Similarly, if $g$ is decreasing:
\begin{equation}
g(-1) \leq 0 \Rightarrow \phi^{'}(\rho_{0})\geq  \frac{\phi(\rho_{0})}{1+\rho_{0}}\Rightarrow  2\phi(\rho_{0}) -(1-\rho_{0}^{2})\phi^{'}(\rho_{0}) \leq 2\phi(\rho_{0})\left[1-\frac{1-\rho_{0}}{2} \right]<0\label{eq:der_decreasing}
\end{equation}
\end{proof}
\vspace*{-0.2cm}
We also utilize  a structural result for convex functions.
\begin{theorem}[\cite{rote1992convergence}]\label{thm:sandwich}
Given $\epsilon>0$, there exists an algorithm that given a univariate convex function $f$ on an interval $[a,b]$ constructs a piecewise linear convex function $\ell$ such that $0\leq f(x)-\ell(x)\leq \epsilon$  for all $x\in[a,b]$ using $O(\sqrt{\frac{(b-a)\Delta}{\epsilon}})$ linear segments where $\Delta = f{'}(b_{-}) - f^{'}(a_{+})$.
\end{theorem}

\subsection{Hashing}
\begin{definition}[Asymmetric Hashing]
Given a set of functions $\mathcal{H}\subset \{h:\mathcal{X}\to \mathcal{U}\}$ and a probability distribution   $\nu$ on $\mathcal{H}\times \mathcal{H}$, we write $(h, g)\sim \mathcal{H}_{\nu}$ to denote a random element   sampled from $\nu$, and call $\mathcal{H}_{\nu}$ a  hashing scheme on $\mathcal{X}$. 
\end{definition}

\begin{definition}[Hash Bucket] Given a finite set $X\subset \mathcal{X}$ and an element $(h,g)\in \mathcal{H}\times \mathcal{H}$,  we define for all $y\in \mathcal{X}$ the hash bucket of $X$ with respect to $y$ as $H_{X}(y):=\{x\in X| h(x)=g(y)\}$.  For such a hash bucket we write $X_{0}\sim H_{X}(y)$ to denote the random variable $X_{0}$ that is uniformly distributed in $H_{X}(y)$ when the set is not empty and equal to $\perp$ when it is.
\end{definition}
 The \emph{collision probability} of  a hashing scheme $\mathcal{H}_{\nu}$ on $\mathcal{X}$ is  defined by $p_{\mathcal{H}_{\nu}}(x,y):=\P_{(h,g)\sim \mathcal{H}_{\nu}}[h(x)=g(y)]$ for all $x,y\in \mathcal{X}$.  Whenever it is clear from the context we will omit $\nu$ from $\mathcal{H}_{\nu}$ and $X$ from $H_{X}(y)$.  We also define $\mathcal{H}^{\otimes k}$ to denote the hashing scheme resulting from stacking $k$ independent hash functions from $\mathcal{H}$. For such hashing schemes we have    $p_{\mathcal{H}^{\otimes k}}(x,y)=[p_{\mathcal{H}}(x,y)]^{k}$ for $x,y\in \mathcal{X}$. 
\vspace*{-0.1in}

\subsection{Multi-resolution Hashing Based Estimators}
We define next the class of estimators that we employ.
\begin{definition}
Given hashing schemes $\mathcal{H}_{1},\ldots,\mathcal{H}_{T}$, with collision probabilities $p_{1},\ldots,p_{T}:\mathcal{X}\times \mathcal{X}\to [0,1]$, and weight  functions $w_{1},\ldots, w_{T}:\mathcal{X}\times \mathcal{X}\to \R_{+}$, we define   for a given set $X\subseteq \mathcal{X}$, the Multi-Resolution Hashing-Based-Estimator   for all $y\in \mathcal{X}$ as:
\begin{equation}
Z_{T}(y):=\frac{1}{|X|}\sum_{t=1}^{T} \frac{w_{t}(X_{t},y )}{p_{t}(X_{t}, y)}|H_{t}(y)|
\end{equation}
where $X_{t} \sim H_{t}(y)=(H_{t})_{X}(y)$ and by setting   $w_{t}(\perp, \cdot)=p_{t}(\perp, \cdot)=1$ for $t\in [T]$. We denote such an estimator by $Z_{T}\sim \mathrm{HBE}_{X}(\{\mathcal{H}_{t},p_{t},w_{t}\}_{t\in [T]})$.
\end{definition}
Again we drop the dependence on $X$ when it is clear from the context. Manipulating  conditional expectations gives us the following basic properties for such estimators.
\begin{lemma}[Moments]\label{lem:moments}
For any $y\in \mathcal{X}$ and $x\in X$ let $T(x,y)=:\{t\in[T]| p_{t}(x, y)>0\}$ and assume that $\forall x\in X,   \sum_{t\in T(x,y)}w_{t}(x,y) = w(x,y) $ for a non-negative function $w:\mathcal{X}\times \mathcal{X}\to \R_{+}$. Then, 
\begin{align}
\E[Z_{T}(y)] &= \mu:= \frac{1}{|X|}\sum_{x\in X}w(x,y)\label{eq:first}\\
\E[Z^{2}_{T}(y)]   &\leq \frac{1}{|X|^{2}}\sum_{x\in X} \left(\sum_{t \in T(x,y)} \frac{w_{t}^{2}(x,y)}{p_{t}(x,y)}  \sum_{z\in X}\frac{\min\{p_{t}(z,y), p_{t}(x,y)\}}{p_{t}(x,y)}\right)  +   \mu^{2}\label{eq:second}
\end{align}
\end{lemma}
The upper bound on the variance comes from  $\E\left[|H_{t}(y)|\bigr| x\in H_{t}(y)\right] \leq \sum_{z\in X}\frac{\min\{p_{t}(z,y), p_{t}(x,y)\}}{p_{t}(x,y)}$.

\subsection{Distance Sensitive Hashing on the unit Sphere}\label{sec:hashing}
\vspace*{-0.1in}
In this  subsection, we describe the hashing scheme of Aumuller et al.~\cite{aumuller2018distance} (see also~\cite{andoni2014beyond,andoni2015optimal})  and give slightly different bounds on the collision probability that are more appropriate for our purposes.  

\paragraph{LSH for unit sphere}We define the hash family $\mathcal{D}_{+}=\mathcal{D}_{+}(t,\zeta)$ that takes as parameters real numbers $t>0$, $\zeta\in(0,1)$ and defines a pair of hash functions $h_{+}:\mathcal{S}^{d-1}\to  [m]\cup \{m+1\}$ and $g_{+}:\mathcal{S}^{d-1}\to [m]\cup\{m+2\}$, where $m$ is given by
\begin{equation}\label{eq:spherical_caps}
m(t,\zeta )= \left\lceil \sqrt{2\pi}(t+1)\log(\frac{2}{\zeta}) e^{\frac{t^{2}}{2}}\right\rceil  
\end{equation}
To define the functions $h_{+},g_{+}$, we sample $m$ normal random vectors $g_{1},\ldots, g_{m}\stackrel{i.i.d.}{\sim}\mathcal{N}(0, I_{d})$ and use them to create $m+2$ hash buckets through the mappings
\begin{align}
h_{+}(x) &:= \min\left( \{i\bigr|\langle x, g_{i}\rangle\geq t \}\cup \{m+1\}\right)\\
g_{+}(x) &:= \min \left( \{i\bigr|\langle x, g_{i}\rangle\geq t \}\cup \{m+2\}\right)
\end{align}
The time and memory required for evaluating the function are both bounded by $O(dm)=O(dt \log(\frac{1}{\zeta})e^{\frac{t^{2}}{2}})$. We also define the hash family $\mathcal{D}_{-}(t,\zeta)$ that is identical to $\mathcal{D}_{+}$ except from the fact that  instead of using $g_{+}$ we use:
\begin{equation}
g_{-}(x) := \min \left( \{i\bigr|\langle x, g_{i}\rangle\geq -t \}\cup \{m+2\}\right)
\end{equation}
The need to use a pair of hash functions arises from the fact that we treat the points in the dataset $X$ and the queries differently. We will write $(h,g)\sim \mathcal{D}_{s}$ for $s\in\{+,-\}$ to indicate such pairs of hash functions. Due to isotropy of the normal distribution the collision probability   only depends on $\langle x,y\rangle$,
\begin{equation}
\P_{(h,g)\sim D_{\pm}}[h(y)=g(x)] = p_{\pm}(\langle x, y\rangle)
\end{equation}
and satisfies $p_{+}(\rho)=p_{-}(-\rho)$ for all $\rho\in[-1,1]$. Utilizing results for Gaussian integrals \cite{szarek1999nonsymmetric,hashorva2003multivariate}, we obtain the following explicit bounds.

\begin{lemma}[Pointwise bounds]\label{lem:collision_basic}
 The collision probability $p_{+}(\rho)$ is  decreasing   and for $\delta>0$ satisfies:
\begin{align}\label{eq:collision_center}
 \frac{\sqrt{2}(1-\zeta)\delta^{2}}{148\sqrt{\pi}} e^{-\frac{1-\rho}{1+\rho}\frac{t^{2}}{2}}&\leq p_{+}(\rho) \leq \frac{2}{\sqrt{\pi}\sqrt{\delta}} e^{-\frac{1-\rho}{1+\rho}\frac{t^{2}}{2}},  \qquad \forall |\rho|\leq 1-\delta\\\label{eq:collision_right}
 \frac{1-\zeta}{2\sqrt{2\pi}(1+\sqrt{2})} e^{-\frac{\delta}{2-\delta} \frac{t^{2}}{2}}&\leq p_{+}(\rho)\leq 1, \qquad \qquad\qquad \hspace{0.5cm}   \forall 1-\delta < \rho \leq 1\\\label{eq:collision_left}
0&\leq p_{+}(\rho) \leq \frac{2}{\sqrt{\pi}\sqrt{\delta}} e^{-\frac{2-\delta}{\delta}\frac{t^{2}}{2}},\qquad    \forall -1\leq \rho \leq -1+\delta
\end{align}    
\end{lemma}
The family $D_{+}$ tends to map correlated points to the same bucket, whereas $D_{-}$ tends to map anti-correlated points together. Combining the two hash families, Aumuller et al.~\cite{aumuller2018distance} created a \emph{Distance Sensitive Hashing} scheme.
\paragraph{DSH for unit sphere}    
Given real numbers $t,\gamma>0$ and $\zeta\in (0,1/2)$, we define the following hash family $\mathcal{D}_{\gamma}(t,\zeta)$ by sampling a $(h_{+},g_{+})\sim \mathcal{D}_{+}(t,\zeta)$ and $(h_{-},g_{-})\sim \mathcal{D}_{-}(\gamma t, \zeta)$. We create the hash functions by $h_{\gamma}(x) :=(h_{+}(x), h_{-}(x))$ and $g_{\gamma}(x):=(g_{+}(x),g_{-}(x))$ and write $(h_{\gamma}, g_{\gamma})\sim \mathcal{D}_{\gamma}(t,\zeta)$. Define the collision probability $p_{\gamma,t}(\rho):=\P_{(h_{\gamma},g_{\gamma})\sim \mathcal{D}_{\gamma}(t,\zeta)}[h_{\gamma}(x)=g_{\gamma}(y)]$.
\begin{corollary}\label{cor:sensitive_bounds}
 Given constants $\gamma,t>0$ and $\zeta \in (0,\frac{1}{2})$ define $t_{\gamma}=t\max\{\gamma,1\}$ a pair of hash functions $(h_{\gamma},g_{\gamma})\sim \mathcal{D}_{\gamma}(t, \zeta)$ can be evaluated using space and time $O(dt_{\gamma}\log(\frac{1}{\zeta})e^{ t_{\gamma}^{2}/2})$. Furthermore, for $\delta>0$ let $C_{1}(\delta):= \left(\frac{148\sqrt{\pi}}{\sqrt{2}(1-\zeta)\delta^{2}}\right)^{2}$ depending only on $\zeta, \delta$ such that:
\begin{align}
\frac{1}{C_{1}} e^{-\left(\frac{1-\rho}{1+\rho}+\gamma^{2}\frac{1+\rho}{1-\rho}\right)\frac{t^{2}}{2}}\leq p_{\gamma,t}(\rho)&\leq C_{1} e^{-\left(\frac{1-\rho}{1+\rho}+\gamma^{2}\frac{1+\rho}{1-\rho}\right)\frac{t^{2}}{2}},  \ &&\forall \ |\rho|\leq 1-\delta\label{eq:tight_sensitive}\\
p_{\gamma,t}(\rho) &\leq  \sqrt{C_{1}(\delta)} e^{-\frac{2-\delta}{\delta}\frac{t_{\gamma}^{2}}{2}}, &&\forall  |\rho|>1-\delta\label{eq:tail_sensitive}
\end{align}
\end{corollary}
 \begin{proof}
As we sample hash functions from the families $D_{+}(t,\zeta)$ and $D_{-}(\gamma t, \zeta)$ independently, the collision probability  $p_{\gamma,t}(\rho)=p_{+}(\rho)p_{-}(\rho)$ is the product of the two collision probabilities. Using Lemma \ref{lem:collision_basic} we get the required statement with $C_{1}(\delta):=\max\left\{ (\frac{\sqrt{2}(1-\zeta)\delta^{2}}{148\sqrt{\pi}} )^{-2}, (\frac{2}{\sqrt{\pi}\sqrt{\delta}} )^{2} \right\}$.
\end{proof}

\section{Variance of Multi-resolution HBE}\label{sec:variance}
In this section, we analyze the variance of Multi-resolution HBE and identify two key design principles: the \emph{$p^{2}$-weighting scheme}, and the \emph{scale-free property} of HBE, for which we give strong theoretical bounds on the variance. Our first step is to obtain a more tractable  bound on \eqref{eq:second}.
\begin{lemma}\label{lem:variance}
 Given   an $n$ point set $X$ and an unbiased $Z_{T}\sim \mathrm{HBE}_{X}(\{\mathcal{H}_{t},p_{t},w_{t}\}_{t\in [T]})$,  there exists explicit  $A\in \R^{n\times n}$ and vector $v\in \R^{n}_{++}$ such that: $\E[Z^{2}_{T}(y)] \leq  \sup_{\|f\|_{1}\leq 1, \|f\|_{v,1}\leq \mu}\{f^{\top} A f\} +   \mu^{2}$.
\end{lemma}
\begin{proof}
Fix $x_{1},\ldots, x_{n}$ potential positions for the $n$ points in the dataset and let $f_{1},\ldots, f_{n}\in [0,1]$ be the fraction of points that are assigned to each of this positions. Moreover for any two positions $x_{i},x_{j}$ let $L_{ij}$ be the set of hash functions such that $p_{t}(x_{i},y)<p_{t}(x_{j},y)$ and $G_{ij}$ be the complement. We get:
\begin{align}
\sum_{j\in [n]}&\frac{\min\{p_{t}(x_{j},y), p_{t}(x_{i},y)\}}{p_{t}(x_{i},y)} \leq  \sum_{j\in [n]}nf_{j}\left( \mathbb{I}[t\in L_{ij}] + \mathbb{I}[t\in G_{ij}]\frac{p_{t}(x_{j},y)}{p_{t}(x_{i},y)}\right)\label{eq:conditional}
\end{align}
Using \eqref{eq:second}, and \eqref{eq:conditional}, the lemma follows by setting  $\nu_{i}=w(x_{i},y)$ and 
\begin{align}
A_{ij}=\sum_{t\in L_{ij}}\frac{w_{t}^{2}(x_{i},y)}{p_{t}(x_{i},y)} + \sum_{t\in G_{ij}}\frac{w_{t}^{2}(x_{i},y)}{p^{2}_{t}(x_{i},y)}p_{t}(x_{j},y)\label{eq:matrix}.
\end{align}
\end{proof}
The main question that the above lemma leaves open, is to how select the functions $\{w_{t}\}$ so that, the estimator is still unbiased, but the variance is minimized.
\subsection{The $p^{2}$-weighting scheme for HBE}
Our goal is to find a set of weights that are only a function of the query $y$ and any point $x$.  To select such a weights we first obtain the following upper bound on \eqref{eq:matrix}
\begin{equation}\label{eq:p2-upper}
\sum_{t\in L_{ij}} \frac{w_{t}^{2}(x_{i},y)}{p_{t}(x_{i},y)} + \sum_{t\in G_{ij}} \frac{w_{t}^{2}(x_{i},y)}{p^{2}_{t}(x_{i},y)}p_{t}(x_{j},y) \leq \sum_{t\in [T]} \frac{w_{t}^{2}(x_{i},y)}{p_{t}^{2}(x_{i},y)}
\end{equation}
The set of weights that minimize \eqref{eq:p2-upper} and for which the HBE is still unbiased  are given by: $w_{t}^{*}(x,y) = \frac{p_{t}^{2}(x,y)}{W(x,y)}w (x,y)$, where  $W(x,y):=\sum_{t\in T(x,y)} p_{t}^{2}(x,y)$. In what follows we denote any unbiased $\mathrm{HBE}_{X}(\{\mathcal{H}_{t},w_{t},p_{t}\}_{t\in [T]})$ with $w_{t}\propto p_{t}^{2}w $ as $\mathrm{HBE}^{2}_{X}(\{\mathcal{H}_{t},p_{t}\}_{t\in [T]})$. 
We aim to quantify precisely how well these estimators can perform by choosing $\{p_{t}\}$ judiciously. To that end, using Lemmas \ref{lem:variance} and  \ref{lem:holder},   we obtain the following   upper bound on the variance.
\begin{theorem}\label{thm:psquared}
 Given a set $X\subseteq S\subset \mathcal{X}$, and $Z_{T}\sim \mathrm{HBE}^{2}_{X}(\{\mathcal{H}_{t},p_{t}\}_{t\in [T]})$ let $p_{*}(x,y):=\sup_{t\in [T]}\{p_{t}(x,y)\}$, then for all $y\in Y$ such that $Z(y)=\mu>0$ and $f_{i}=f(x_{i}):=\min\{1,\frac{\mu}{w(x_{i},y)}\}$, we get:
\begin{align}
\E[Z^{2}_{T}(y)] \leq \mu^{2}+ 4 \sup_{x_{1},x_{2}\in S}\left\{f^{2}_{1}\frac{w^{2}(x_{1},y)}{p_{*}(x_{1},y)}  +f^{2}_{2}\frac{w^{2}(x_{2},y)}{p_{*}(x_{2},y)}  + f_{1}f_{2} \left(\frac{w^{2}(x_{1},y)}{p_{*}^{2}(x_{1},y)}
 +\frac{w^{2}(x_{2},y)}{p_{*}^{2}(x_{2},y)} \right)D_{T}(x_{1},x_{2})\right\} \nonumber
\end{align}
where $D_{T}(x_{1},x_{2}):=\max_{t\in [T]}\min\{p_{t}(x_{1},y),p_{t}(x_{2},y)\}\leq \min\{p_{*}(x_{1},y), p_{*}(x_{2},y)\}$.
\end{theorem}
\begin{proof}
 Using Lemma \ref{lem:holder} we get
\begin{align}
\sup_{\|f\|_{w,1}\leq \mu, \|f\|_{1}\leq 1}f^{\top}Af \leq 4 \sup_{ij}\left\{\min\{1,\frac{\mu}{w(x_{i},y)}\} |A_{ij}|\min\{1,\frac{\mu}{w(x_{j},y)}\}  \right\}
\end{align}
with $A_{ij}= \left(\sum_{t\in L_{ij}}\frac{w_{t}^{2}(x_{i},y)}{p_{t}(x_{i},y)} + \sum_{t\in G_{ij}}\frac{w_{t}^{2}(x_{i},y)}{p^{2}_{t}(x_{i},y)}p_{t}(x_{j},y) \right)$. Setting $f_{i}=\min\{1,\frac{\mu}{w(x_{i},y)}\}$ and 
 $\tilde{A}_{ij} = f_{i} |A_{ij}|f_{j}$, we get by the above $\sup_{\|f\|_{w,1}\leq \mu, \|f\|_{1}\leq 1}f^{\top}Af \leq 4 \sup_{ij}\left\{\tilde{A}_{ii}+\tilde{A}_{jj}+\tilde{A}_{ij}+\tilde{A}_{ji} \right\}$. 
Let $V_{ij}$ be the expression in brackets. For the $p^{2}$-weighting scheme $w_{t}(x,y)= \frac{p_{t}^{2}(x,y)}{W(x,y)}w(x,y)$ we get
\begin{align}
V_{ij} &= f_{i}^{2}\sum_{t\in L_{ij}}\frac{w^{2}(x_{i},y)}{W^{2}(x_{i},y)}p_{t}^{3}(x_{i},y) + f_{j}^{2} \sum_{t\in G_{ij}}\frac{w^{2}(x_{j},y)}{W^{2}(x_{j},y_)}p_{t}^{3}(x_{j},y) \\
&\qquad +f_{i}f_{j} \left(\sum_{t\in L_{ij}}\frac{w^{2}(x_{i},y)}{W^{2}(x_{i},y)}p_{t}^{3}(x_{i},y) + \sum_{t\in G_{ij}}\frac{w^{2}(x_{i},y)}{W^{2}(x_{i},y)}p_{t}^{2}(x_{i},y)p_{t}(x_{j},y) \right)\nonumber \\
&\qquad+ f_{j}f_{i} \left(\sum_{t\in L_{ji}}\frac{w^{2}(x_{j},y)}{W^{2}(x_{j},x)}p_{t}^{3}(x_{j},y) + \sum_{t\in G_{ji}}\frac{w^{2}(x_{j},y)}{W^{2}(x_{j},y)}p_{t}^{2}(x_{j},y)p_{t}(x_{i},y)\right)\nonumber
\end{align}
 Using $W(x,y)=\sum_{t\in [T]}p_{t}^{2}(x,y)\geq p_{*}^{2}(x,y)$ and $p_{t}\leq p_{*}(x,y)$
\begin{align}
V_{ij}&\leq f_{i}^{2}\frac{w^{2}(x_{i},y)}{W^{2}(x,y_{i})}p_{*}(x_{i},y)\sum_{t\in[T]}p_{t}^{2}(x_{i},y) + f_{j}^{2} \frac{w^{2}(x_{j},y)}{W^{2}(x_{j},y)}p_{*}(x_{j},y)\sum_{t\in [T]}p_{t}^{2}(x_{j},y)  \\
&\qquad +f_{i}f_{j} \frac{w^{2}(x_{i},y)}{W^{2}(x_{i},y)}\sum_{t\in [T]}p_{t}^{2}(x_{i},y)  \max\left\{\max_{t\in L_{ij}}p_{t}(x_{i},y), \max_{t\in G_{ij}}p_{t}(x_{j},y) \right\}\nonumber\\
&\qquad +f_{j}f_{i} \frac{w^{2}(x_{j},y)}{W^{2}(x_{j},y)}\sum_{t\in[T]}p_{t}^{2}(x_{j},y)  \max\left\{\max_{t\in L_{ji}}p_{t}(x_{j},y), \max_{t\in G_{ji}}p_{t}(x_{i},y) \right\}\nonumber\\
&\leq f_{i}^{2}\frac{w^{2}(x_{i},y)}{p_{*}(x_{i},y)}  +f_{j}^{2}\frac{w^{2}(x_{j},y)}{p_{*}(x_{j},y)} \\
&\qquad + f_{i}f_{j} \left(\frac{w^{2}(x_{i},y)}{p_{*}^{2}(x_{i},y)}+\frac{w^{2}(x_{j},y)}{p_{*}^{2}(x_{j},y)} \right)\max\left\{\max_{t\in L_{ij}}p_{t}(x_{i},y), \max_{t\in G_{ij}}p_{t}(x_{j},y) \right\}\nonumber
\end{align}
Since $G_{ji}\subseteq L_{ij}$ and vice versa, setting $D_{T}(x_{i},x_{j}):= \max\left\{\max_{t\in L_{ij}}p_{t}(x_{i},y), \max_{t\in L_{ji}}p_{t}(x_{j},y) \right\}$
 we arrive at the following bound on:
\begin{equation}
V_{ij} \leq f_{i}^{2}\frac{w^{2}(x_{i},y)}{p_{*}(x_{i},y)}  +f_{j}^{2}\frac{w^{2}(x_{j},y)}{p_{*}(x_{j},y)} + f_{i}f_{j} \left(\frac{w^{2}(x_{i},y)}{p_{*}^{2}(x_{i},y)}+\frac{w^{2}(x_{j},y)}{p_{*}^{2}(x_{j},y)} \right)D_{T}(x_{i},x_{j})
\end{equation}
To complete the proof we show the following:
\begin{align*}
D_{T}(x_{1},x_{2}) &= \max\left\{\max_{t\in L_{12}}p_{t}(x_{1},y), \max_{t\in L_{21}}p_{t}(x_{2},y) \right\}\\
&=\max\left\{\max_{t\in L_{12}}\min\{p_{t}(x_{1},y),p_{t}(x_{2},y)\}, \max_{t\in L_{21}}\min\{p_{t}(x_{1},y),p_{t}(x_{2},y)\}\right\}\\
&= \max_{t} \min\{p_{t}(x_{1},y),p_{t}(x_{2},y)\}
\end{align*}
Noticing that $ \max_{t} \min\{p_{t}(x_{1},y),p_{t}(x_{2},y)\}\leq p_{*}(x_{1},y)$ and $ \max_{t} \min\{p_{t}(x_{1},y),p_{t}(x_{2},y)\}\leq p_{*}(x_{2},y)$,  we get the statement.
\end{proof}

\subsection{Scale-free  Multi-Resolution Hashing}
The development above has revealed that the crucial parameter for consideration of $\mathrm{HBE}^{2}$ is the \emph{pointwise maximum hashing probability} $p_{*}(x,y)$. Here, we analyze a specific family of estimators where $p_{*}(x,y)$ has polynomial dependence with $w(x,y)$.
\begin{definition} Given $M\geq 1$, $\beta\in [0,1]$ and function $w$, an estimator $Z_{T}\sim \mathrm{HBE}_{X}^{2}(\{\mathcal{H}_{t},p_{t}\}_{t\in [T]})$ is called $(\beta, M)$-\emph{scale free}, if $M^{-1}\cdot w^{\beta}(x,y)\leq p_{*}(x,y)\leq M\cdot w^{\beta}(x,y)$ for all $x\in X$ and $y\in \mathcal{X}$.
\end{definition}
Exploiting the scale-free property we get explicit bounds on the variance.
\begin{theorem}[Scale-free]\label{thm:scale-free} Let $Z_{T}\sim  \mathrm{HBE}_{X}^{2}(\{\mathcal{H}_{t},p_{t}\}_{t\in [T]})$ be a $(\beta, M)$-scale free estimator, then:
\begin{equation*}
\E[Z^{2}_{T}(y)]\leq V_{\beta,M}(\mu):=8M^{3} \mu^{2} \left[\frac{1}{\mu^{\beta}}+\frac{1}{\mu^{1-\beta}} \right] +  \mu^{2}
\end{equation*}
\end{theorem}
Our theorem shows that the optimal worst-case variance is achieved for $\beta^{*}=1/2$ and improves over uniform random sampling by a factor of $O(\frac{1}{\sqrt{\mu}})$. A theorem of similar nature but with a more involved proof was given in \cite{charikar2017hashing} for $\beta \in [\frac{1}{2},1]$. 

\begin{proof}
  For $i\in[1,2]$ let $w_{i}:=w(x_{i},y)$ and $f_{i}$ as in Theorem \ref{thm:psquared}. Using the scale-free property, Theorem \ref{thm:psquared} and $D_{T}(x_{1},x_{2})\leq \min\{p_{*}(x_{1},y),p_{*}(x_{2},y)\}$  we arrive at:
\begin{align}
\E[Z_{T}^{2}] &\leq \mu^{2} + 4 M^{3} \sup_{x_{1},x_{2}\in S}\left\{f^{2}_{1}w_{1}^{2-\beta}  +f^{2}_{2}w_{2}^{2-2\beta}  + f_{1}f_{2} \left(w_{1}^{2-2\beta}
 +w_{2}^{2-2\beta}\right)\min\{w_{1},w_{2}\}^{\beta}\right\} \nonumber
\end{align}
Due to the definition of $f_{i}$ the last expression is only a function of $w_{1},w_{2}$ and solving the optimization problem boils down to a case analysis. We focus on the case $w_{1}\geq \mu$, $w_{2}\leq \mu$,  for which the expression in the parenthesis becomes:
\begin{align} 
 \mu^{2} w_{1}^{-\beta} + w_{1}^{1-2\beta}w_{2}^{2\beta}\mu + w_{2}^{2-\beta}w_{1}^{\beta}\mu + w_{2}^{2-\beta}
\end{align}
The  weights that maximize the expression are $w_{1}^{*}=1$ and $w_{2}^{*}=\mu$. $\mu^{2} + \mu^{1+2\beta} + \mu^{1+\beta} + \mu^{2-\beta} \leq 2\mu^{2}[\mu^{-\beta}+\mu^{\beta-1}]$. The other cases $w_{1},w_{2}\leq \mu$ and $w_{1},w_{2}\geq \mu$ follow similarly.
\end{proof} 
\section{Approximation of Convex Functions}\label{sec:family}
In this section, we show how to use the logarithm  $h_{\gamma,t}(\rho)$, given below,  of the \emph{idealized hashing probability} of the Distance Sensitive Hashing scheme  to construct a set of functions whose supremum approximates any non-positive convex Lipschitz function $\phi(\rho)$.
\begin{equation}
h_{\gamma,t}(\rho):=-\left(\frac{1-\rho}{1+\rho}+\gamma^{2}\frac{1+\rho}{1-\rho}\right)\frac{t^{2}}{2} \label{eq:idealized-2}
\end{equation}
Some basic properties of this family of functions are given below. 
\begin{proposition}[Concavity]\label{prop:derivatives}
 For $\gamma\geq 0$, the function $h_{\gamma,t}$ attains its maximum at $\rho^{*}(\gamma) = \frac{1-\gamma}{1+\gamma}$ and
\begin{itemize}
\item [(a)]If $0\leq \gamma\leq 1$, the function is concave for all $\rho\in [\rho^{*}(\gamma^{\frac{2}{3}}), 1]$ and $\rho^{*}(\gamma) \geq \rho^{*}(\gamma^{\frac{2}{3}})$ holds.
\item [(b)]If $\gamma \geq 1$, the function is concave for all $\rho\in [-1, \rho^{*}(\gamma^{\frac{2}{3}})]$ and   $\rho^{*}(\gamma) \leq \rho^{*}(\gamma^{\frac{2}{3}})$ holds.
\end{itemize} 
\end{proposition}
The above properties will be used to show that, by picking parameters $\gamma_{0},t_{0}$ appropriately, if we approximate the convex function $\phi$  \emph{locally} at some point $\rho_{0}\in [-1,1]$ up to \emph{first order} (value and derivative),  then  $h_{\gamma_{0},t_{0}}(\rho)\leq \phi(\rho)$ for all $\rho\in [-1,1]$. Thus even a single hash function is sufficient to provide a lower bound. Most of the work is devoted to show that we can get a good \emph{upper bound} on $\phi$ using a small number of functions to approximate $\phi$ locally at a set of interpolation points $\rho_{1},\ldots,\rho_{T}$.  We define the following parametrization.  Given $\delta>0$ for $|\rho_{0}|\leq 1 - \delta$, let
\begin{align}\label{eq:gamma}
\gamma_{0}^{2} &:= \left(\frac{1-\rho_{0}}{1+\rho_{0}}\right)^{2}\frac{2 \phi(\rho_{0})+(1-\rho_{0}^{2})\phi^{'}(\rho_{0})}{2\phi(\rho_{0})-(1-\rho_{0}^{2})\phi^{'}(\rho_{0})}\\
t_{0}^{2}&:= -\frac{1}{2}\frac{1+\rho_{0}}{1-\rho_{0}}\left[2\phi(\rho_{0})-(1-\rho_{0}^{2})\phi^{'}(\rho_{0})\right]\label{eq:threshold}
\end{align}
and for  fixed $\phi$ and $\rho_{0}\in [-1+\delta, 1-\delta]$ define $h_{\rho_{0}}(\rho):=h_{\gamma_{0},t_{0}}(\rho)$ 
. This  parametrization is well defined due to  Lemma \ref{lem:positive}. For $\rho_{0}\in\{-1,+1\}$ (boundary) we define $h_{\pm 1}(\rho):= -\frac{1\mp \rho}{1\pm \rho}\frac{t_{\pm 1}^{2}}{2} + \phi(\pm 1)$, 
where $t^{2}_{\pm 1} = 4\max\{\pm  \phi^{'}(\pm1), 0\}$. Under our assumptions $\phi\leq 0$, hence  the constant term above can be implemented by sub-sampling the data set with probability $e^{\phi(\pm 1)}$. The following bounds on the parameters $\gamma_{0},t_{0}$ will be useful.
\begin{corollary}[Complexity]\label{col:exponent_bounds}
 Under the conditions of Lemma \ref{lem:positive}, we have the following bounds: $t_{0}^{2} \leq - 2\frac{1+\rho_{0}}{1-\rho_{0}}\phi(\rho_{0})$, $t_{0}^{2}\gamma_{0}^{2} \leq -2\frac{1-\rho_{0}}{1+\rho_{0}}\phi(\rho_{0})$, and $t_{0}^{2}\max\{\gamma_{0}^{2},1\}\geq- \frac{1+\rho_{0}^{2}}{1-\rho_{0}^{2}}\phi(\rho_{0})$.
\end{corollary}

Using this family of functions we show we can approximate a convex function arbitrarily well.

\begin{theorem}[Approximation] \label{thm:convex_approx}
Given $\epsilon>0$, for every convex function $\phi$ there exists a set $\mathcal{T}_{\epsilon}(\phi)\subset [-1,1]$ of size $O\left(\sqrt{\frac{L(\phi)}{\epsilon}}\log(\frac{L(\phi)}{\epsilon})\right)$ such that $0\leq\phi(\rho) -  \sup_{\rho_{0}\in \mathcal{T}_{\epsilon}}\{h_{\rho_{0}}(\rho)\} \leq 2\epsilon$ for all $\rho\in [-1,1]$.
\end{theorem}
\subsection{Proof of Approximation Theorem}
To prove the above theorem it is sufficient, due to Theorem \ref{thm:sandwich}, to only show how to \emph{approximate linear functions}.  For $\rho$ away from $\{-1,1\}$, this is done in Lemma \ref{lem:linear}, where the \emph{interpolation points are given explicitly}. Lemma \ref{lem:boundary} treats the case near the boundary.  By symmetry of the family of hash functions we only need to show our result for $[-1,0]$.

\begin{lemma}\label{lem:linear}
 Let $\ell$ be a linear function on $[\rho_{-},\rho_{+}]\subseteq [-1+\delta,0]$. Given $\epsilon>0$, let $T =  \lfloor \frac{\log(\frac{1-|\rho_{+}|}{1-|\rho_{-}|})}{\log(1+\sqrt{\frac{\epsilon}{8|\ell_{\min}|}})}\rfloor$ and define $\rho_{i}:= \rho_{-} + (1-|\rho_{-}|) \left[\left(1+\sqrt{\frac{\epsilon}{8R(\ell)}}\right)^{i}-1 \right]$ for $i=0,\ldots, T$. Then, for all $\rho\in [\rho_{-},\rho_{+}]$ there exists $i(\rho)\in [T]\cup\{0\}$ such that $0\leq \ell(\rho)-h_{\rho_{i(\rho)}}(\rho)\leq \epsilon$.
\end{lemma}

\begin{lemma} \label{lem:boundary}
Given $\epsilon>0$, let $\delta(\epsilon):=\min\{1, \sqrt{\frac{\epsilon}{4L(\phi)}},\frac{\epsilon}{L(\phi)}\}$. Then  $0 \leq \phi(\rho) -h_{-1}(\rho)\leq \epsilon$ for all $\rho$ in the interval $[-1, -1+\delta(\epsilon)]$.
\end{lemma}
\begin{proof}
 If $\phi^{'}(-1)\geq 0$, then $0\leq \phi(\rho)-h(\rho)=\phi(\rho)-\phi(-1)\leq L(\rho+1) \leq L\delta$. If $\phi^{'}(-1)<0$ then  by the Taylor remainder theorem and $0\leq \delta\leq1$ we get
\[
0\leq \phi(\rho)-h_{-1}(\rho) \leq \frac{1}{2}\frac{2}{(2-\delta)^{3}}4|\phi^{'}(-1)|\delta^{2}\leq 4L\delta^{2}
\]
Using the definition of $\delta(\epsilon)$ we get the statement.
\end{proof}
The previous lemmas provide only local approximation to the function.  Proposition \ref{prop:infimum} below is used to show that the functions we construct are a lower bound to the piecewise linear approximation on the whole interval $\rho\in [-1,1]$, which in turn implies a lower bound for the function $\phi(\rho)$.

 \begin{proposition}\label{prop:infimum}
Let $\phi:[-1,1]\to \R$ be an non-decreasing (resp non-increasing) convex function and $g:[-1,1]\to \R$ a function that attains a global maximum at $\rho^{*}$, is concave in $[-1, \rho^{*}]$ (resp $[\rho^{*},1]$), and $\exists \rho_{0}\in [-1,\rho^{*}]$ (resp. $[\rho^{*},1]$)  such that $\phi'(\rho_{0})=g^{'}(\rho_{0})$, then $\inf_{\rho\in [-1,1]}\{\phi(\rho)-g(\rho)\} = \phi(\rho_{0})-g(\rho_{0})$.
\end{proposition}

\begin{proof}[Proof of Theorem \ref{thm:convex_approx}]
Given $\epsilon>0$, let $\delta(\epsilon)$ as in Lemma \ref{lem:boundary}. We start by applying Theorem \ref{thm:sandwich} separately on the function $\phi$ restricted on the interval $[-1+\delta,0]$ and $\phi$ restricted on $[0,1-\delta]$ to get piecewise linear convex approximation $\ell$ to $\phi$ such that $0\leq \phi(\rho)-\ell(\rho)\leq \epsilon$ for all $|\rho|\leq 1-\delta$. Let $I^{-}=\{[\rho^{-}_{j-1},\rho^{-}_{j}]\}_{j \in [J^{-}]}$ and $I^{+}=\{[\rho^{+}_{j-1},\rho^{+}_{j}]\}_{j \in [J^{+}]}$  with $J^{\pm}= O(\sqrt{\frac{L(\phi)}{\epsilon}})$  be the corresponding decompositions of $[-1+\delta,0]$ and $[0,1-\delta]$ in contiguous subintervals where the function $\ell$ is linear.  For each $j\in [J^{\pm}]$, let $\mathcal{T}^{\pm}_{j}$ be the set of points  resulting  by applying Lemma \ref{lem:linear}  to $[\rho^{\pm}_{j-1},\rho^{\pm}_{j}]$ and set $T^{\pm}_{j}=|\mathcal{T}_{j}^{\pm}|$.  We define the following set of points $\mathcal{T}_{\epsilon}(\phi):=\left(\cup_{j=1}^{J^{+}}\mathcal{T}^{+}_{j}\right) \cup \left(\cup_{j=1}^{J^{-}}\mathcal{T}^{-}_{j}\right)\cup \{1,-1\}$. We have
\begin{align*}
|\cup_{j=1}^{J_{\pm}}\mathcal{T}^{\pm}_{j}| & \leq \sum_{j=1}^{J^{\pm}}(1+T^{\pm}_{j}) \leq J^{\pm}+\frac{\log (\frac{1}{\delta})}{{\log(1+\sqrt{\frac{\epsilon}{8R(\phi)}})}}
\end{align*}
Using $\log(1+x)\geq \frac{2}{3}x$ for $x\in [0,1]$ and $R(\phi)\leq 2L(\phi)$, we get that $|\mathcal{T}_{\epsilon}(\phi)|=O\left(\sqrt{\frac{L(\phi)}{\epsilon}}\log(\frac{L(\phi)}{\epsilon})\right)$. 

Let $\hat{\phi}(\rho) := \sup_{\rho_{0}\in \mathcal{T}_{\epsilon}(\phi)}\{h_{\rho_{0}}(\rho)\}$. Due to Propositions \ref{prop:derivatives} and \ref{prop:infimum}, we get    $\phi(\rho)\geq\ell(\rho)\geq  h_{\rho_{0}}(\rho)$ for all $\rho$ and $\rho_{0}\in \mathcal{T}_{\epsilon}(\phi)$ and consequently $\phi(\rho)-\hat{\phi}(\rho)\geq 0$. Let $T=|\mathcal{T}_{\epsilon}(\phi)|$  and $\rho_{1},\ldots, \rho_{T}$ an increasing ordering of points in $\mathcal{T}_{\epsilon}(\phi)$. We have
\begin{align}
\sup_{\rho\in [-1,1]}\{\phi(\rho)-\hat{\phi}(\rho)\} = \max_{i\in [T-1]}\sup_{\rho\in [\rho_{i},\rho_{i+1}]}\{\phi(\rho)-\hat{\phi}(\rho)\}\leq \max_{i\in [T-1]}\sup_{\rho\in [\rho_{i},\rho_{i+1}]}\left\{\phi(\rho)-\max\{h_{\rho_{i}}(\rho),h_{\rho_{i+1}}(\rho) \}\right\}\nonumber
\end{align}
which is bounded by $2\epsilon$ due to Theorem \ref{thm:sandwich} and Lemmas \ref{lem:boundary}, \ref{lem:linear}. 
\end{proof}
\section{Scale-free Multi-Resolution Hashing for Log-convex functions}\label{sec:scale-free}
In the previous section, we have shown that using the idealized hashing probabilities one can approximate a log-convex function up to arbitrary multiplicative accuracy. In this section, we use this fact to construct explicit scale-free Multi-resolution HBE, that constitutes the main ingredient needed to prove our main result.
\begin{theorem}\label{thm:log-convex}
 Given a convex function $\phi$, $X\subset \mathcal{S}^{d-1}$ and $\beta\in [0,1]$,  there exist an explicit constant $M_{\phi}$ and $(\beta, M_{\phi})$-scale free estimator $Z_{T}\sim \mathrm{HBE}_{X}^{2}(\{\mathcal{H}_{t},p_{t}\}_{t\in [T]})$ for $Z_{\phi}(y)$ with complexity $O(d\{L(\phi)\}^{5/6}M_{\phi})$.
\end{theorem}
\begin{proof} 
The main challenge in proving the result is to trade-off  \emph{complexity of evaluating} the hashing scheme versus the \emph{fidelity of the approximation} of $\beta [\phi(\langle x,y\rangle)-\phi_{\max}]$ by $\log p_{*}(x,y)$ that affects the variance. In order to do that, set $\delta^{*}=\frac{1}{2\beta L(\phi)}$ and for $C^{*}=C_{1}(\delta^{*})$ as in Corollary \ref{cor:sensitive_bounds}, define 
\begin{equation}\label{eq:power}
k^{*}= \left\lceil\left\{\frac{2\beta^{2}}{\log C^{*}}L(\phi)R(\phi)\right\}^{1/3} \right\rceil
\end{equation}
We further define a ``smoothed" version of $\phi$ as $\tilde{\phi}(\rho):= \frac{\beta(\phi(\rho)-\phi_{\max})}{k^{*}}$. If $L(\tilde{\phi})=\frac{\beta}{k^{*}}L(\phi)<2$ then  the variation in the function $R(\tilde{\phi})<4$ is too small and a constant number of random samples suffice to answer any query. So, we only deal with the interesting case when   and $L(\tilde{\phi})\geq 2$ and $R(\tilde{\phi})\geq 4$.
\begin{enumerate}
\item \emph{Approximation:} let $\mathcal{T}_{1/2}=\mathcal{T}_{1/2}(\tilde{\phi})$ be the set of interpolation points resulting from invoking Theorem \ref{thm:convex_approx} for $\tilde{\phi}$  and  $\epsilon=\frac{1}{2}$. For this set of points we have  $\Bigl|\sup_{\rho_{0}\in \mathcal{T}_{1/2}}\{h_{\rho_{0}}(\rho)\} - \tilde{\phi}(\rho) \Bigr| \leq 1$.\label{it:approx}
\item  \emph{Hashing scheme:} let $\rho_{1}<\ldots<\rho_{T}$ be an increasing enumeration of points in $\mathcal{T}_{1/2}$. For each $t\in [T]$, let  $\tilde{\mathcal{H}}_{t}$ be the DSH family with collision probability $\tilde{p}_{t}$ and  parameters given by \eqref{eq:gamma} and \eqref{eq:threshold} (for $\tilde{\phi}$ and $\rho_{t}$). We raise each hashing scheme to the $k^{*}$-th power to get $\mathcal{H}_{t}:=\tilde{\mathcal{H}}^{\otimes k^{*}}_{t}$ with collision probability $p_{t}:=\tilde{p}_{t}^{k^{*}}$.  Using Lemma \ref{lem:collision_basic} and Corollary \ref{cor:sensitive_bounds}\label{it:fidelity} we show:
\begin{lemma}\label{lem:fidelity}
$\Bigl|\sup_{t\in[T]}\{\log p_{t}(\rho)\}-k^{*}\sup_{t\in [T]}\{h_{\rho_{t}}(\rho)\}\Bigr| \leq   k^{*}\log C_{1} $ for all $\rho\in [-1,1]$.
\end{lemma}
\item \emph{Scale-free property:} by the previous two steps and noting that $\log w^{\beta}(x,y) = k^{*}\tilde{\phi}(\langle x,y\rangle)$
\begin{equation}
\Bigl|\sup_{t\in[T]}\{\log p_{t}(\rho)\}-\log w(x,y)^{\beta}\Bigr| \leq k^{*}+k^{*}\log C_{1} \leq 2k^{*}\log C_{1}\label{eq:below}
\end{equation}
This shows that $Z_{T}\sim \mathrm{HBE}_{X}^{2}(\{\mathcal{H}_{t},p_{t}\}_{t\in[T]})$ is $(\beta,M_{\phi})$-scale free with $M_{\phi}:=e^{2k^{*}\log C_{1}}$.
\item \emph{Complexity:} To bound the complexity of the estimator  $Z_{T}\sim \mathrm{HBE}_{X}^{2}(\{\mathcal{H}_{t},p_{t}\}_{t\in[T]})$, we need by \eqref{eq:spherical_caps}, \eqref{eq:gamma},  \eqref{eq:threshold}  to bound $t_{\gamma_{0}}^{2}=t^{2}_{0}\max\{\gamma_{0}^{2},1\}$ for $\rho_{0}\in \mathcal{T}_{1/2}(\tilde{\phi})$. Using Corollary \ref{col:exponent_bounds} we get
\begin{lemma}\label{lem:complexity} If $L(\tilde{\phi})\geq 2$ and $R(\tilde{\phi})\geq \frac{1}{2}$, then $\forall \rho_{0}\in \mathcal{T}_{1/2}(\tilde{\phi})$, $t_{\gamma_{0}}^{2}   \leq 8\left(\frac{\beta}{k} \right)^{2} L(\phi) R(\phi)$.
\end{lemma}
Hence, the  complexity of evaluating the estimator is 
$O\left(|\mathcal{T}_{\frac{1}{2}(\tilde{\phi})}|k^{*}d\log(\frac{1}{\zeta})e^{4\left(\frac{\beta}{k^{*}}\right)^{2}L(\phi) R(\phi)}\right)$, by Theorem \ref{thm:convex_approx} and our choice   \eqref{eq:power}, this is bounded by $O(dL(\phi)^{5/6}M_{\phi})$.
\end{enumerate}
\end{proof}

\subsection{Main Result}
\begin{theorem}\label{thm:main}
 Given $\epsilon,\tau \in (0,1)$, for every convex function $\phi$ with Lispchitz constant $L(\phi)$, there exists an explicit constant $M_{\phi}$ and a data structure using space $O(dL(\phi)^{5/6}M^{3}_{\phi}\frac{1}{\epsilon^{2}}\frac{1}{\sqrt{\tau}}\cdot n)$ and query time $O(dL(\phi)^{5/6}M^{4}_{\phi}\frac{1}{\epsilon^{2}}\frac{1}{\sqrt{\mu}})$ that for any $y\in \mathcal{S}^{d-1}$ with constant probability  can either produce an $(1+ \epsilon)$ approximation to $\mu=Z_{\phi}(y)\geq \tau$ or assert  that $\mu<\tau$.
\end{theorem}

\begin{proof} Follows by invoking Theorems \ref{thm:log-convex}, \ref{thm:scale-free} and \ref{thm:v-bounded}  for $\beta^{*}=1/2$.
\end{proof}
The explicit constant $M_{\phi}:=e^{\{2\log(C^{*})\sqrt{L(\phi)R(\phi)}\}^{2/3}}$ (where $R(\phi)\leq 2L(\phi)$ is the range of $\phi$ and $\log(C^{*})=O(\log L(\phi))$)  is sub-exponential in $L(\phi)$ and is of similar nature to the evaluation time of the Andoni-Indyk LSH~\cite{andoni2006near} and Spherical LSH~\cite{andoni2015optimal}. It corresponds to the  number of randomly placed spherical caps  of certain size that are required to  cover most of the unit sphere. 
\begin{proof}[Proof of Theorem \ref{thm:simple}] The simplified version of our main result follows by setting $L\leq (1-\delta) \log n$. We have that $\mu \geq e^{-2L(\phi)}\geq n^{2(1-\delta)}\Rightarrow \frac{1}{\sqrt{\mu}}\leq n^{1-\delta}$ and $L(\phi)^{5/6}M^{4}_{\phi}=e^{O(\log^{2/3}(n)\log \log n)}=n^{o(1)}$. 
\end{proof}

\section{Reduction from Euclidean Space to Unit Sphere}\label{sec:reduction}
In order to extend our method from unit sphere to bounded subsets of Euclidean space the main observation is that given $\gamma\in (0,1]$, if for two sets $S_{x}, S_{y}\subset \R^{d}$ we have that 
 $\forall x_{1},x_{2}\in S_{x}, \|x_{1}\|/\|x_{2}\|\leq (1+\gamma)$ and $\forall y_{2},y_{1}\in S_{y}$, $\|y_{1}\|/\|y_{2}\|\leq (1+\gamma)$, then $\forall x_{1},x_{2}\in S_{x}, \forall y_{1},y_{2}\in S_{y}$
\begin{align}\label{eq:approximation}
  \langle x_{1}, y_{1}\rangle \approx \|x_{2}\|\|y_{2}\| \left\langle \frac{x_{1}}{\|x_{1}\|}, \frac{y_{1}}{\|y_{1}\|}\right\rangle  .
\end{align}
This fact suggests the following strategy:
\begin{enumerate}
\item Partition the data set $X=X_{1}\uplus\ldots\uplus X_{K}$ and the set of possible queries $Y=Y_{1}\uplus\ldots\uplus Y_{K}$  in \emph{spherical annuli} $\{X_{i}\}_{i\in[K]}$ and $\{Y_{j}\}_{j\in[K]}$.
\item For each pair $(X_{i},Y_{j})$ use the approximation \eqref{eq:approximation} and assume that for some $r_{i}$ and $r_{j}$ all points in $X_{i}$ and $Y_{j}$  approximately lie on $r_{i}\mathcal{S}^{d-1}$ and $r_{j}\mathcal{S}^{d-1}$ respectively.
\item For each such pair construct a Multi-resolution HBE to obtain a low-variance unbiased estimator of the contribution of points in $X_{i}$ for any possible value of $j\in [K]$ (annulus the query might belong to).
\item Sum up the contribution for all $i\in [K]$ to obtain the final estimator and bound its variance.
\end{enumerate}
Our approach applies to the following general class of functions:
\begin{equation}\label{eq:general_form}
w(x,y) = p_{0}(\|x\|)e^{\phi(\langle x,y \rangle)+\mathcal{A}(y)}
\end{equation}
where $\phi$ is convex and Lipschitz, $\mathcal{A}(y)$ arbitrary\footnote{For any given   query $y$, $e^{\mathcal{A}(y)}$ is a constant factor that can be factored out.} and $p_{0}:\R_{++}\to \R_{++}$ satisfies a  notion of smoothness that is related to Lipschitz continuity under   the \emph{Hilbert metric} $d_{H}(x,y):= |\log(\frac{x}{y})|$  for $x,y\in \R_{+}$. 

\begin{definition}\label{def:log-lipshcitz}
 For $H,\delta\geq 0$ and $\gamma\in (0,1]$, a function $p_{0}:\R_{++}\to \R_{++}$ is called $(H,\delta,\gamma)$-\emph{log-Lipschitz}, if for all $r_{1}\geq r_{2}>0$ such that $r_{1}\leq (1+\gamma)r_{2}$ we have $\left|\log\left(p_{0}(r_{1})/p_{0}(r_{2})\right) \right| \leq H\cdot\gamma   + \delta$.
\end{definition}
This notion of smoothness implies that the function \emph{changes multiplicatively within each annulus}.

\begin{proposition}\label{prop:log-parameters}
For $\gamma\in (0,1]$ and all $r\in (0, R]$ the function $r^{q}e^{f(r)}$ is $(|q|,\  L(f)R\gamma, \ \gamma)$-log-Lipschitz. 
\end{proposition}
\begin{proof}Let $r_{1},r_{2}\in (0,R]$ such that $r_{2}\leq r_{1}\leq (1+\gamma)r_{2}$,  then
\begin{equation}
\left|\log\left(r_{1}^{q}e^{f(r_{1})}/(r_{2}^{q}e^{f(r_{2})})\right)\right| \leq |q| |\log(r_{1}/r_{2})| + |f(r_{1})-f(r_{2})| \leq |q| \gamma + L(f) R \gamma.
\end{equation} 
\end{proof}
Functions that are of the form \eqref{eq:general_form} include the Gaussian kernel $e^{-\|x-y\|^{2}}$ or the norm of the derivative of the logistic log-likelihood $\|\nabla_{y}\log(1+\exp(\langle x, y\rangle))\|=\|x\|(1+e^{-\langle x , y\rangle})^{-1}$.
For concreteness we are going to assume that the function $p_{0}$ is $(q, HR\gamma, \gamma)$-log-Lipshcitz for some $q,H>0$, as in Proposition \ref{prop:log-parameters}, instead of using general $\delta$ as in Definition \ref{def:log-lipshcitz}. However, our result applies also to the more general case.  In the rest of this section, we carry out the   strategy outlined above. 
\subsection{Partitioning in Spherical Annuli}\label{ssec:paritioning}
Given $0<\gamma\leq 1$, a dataset $X$ and a set of possible queries $Y$, define
\begin{align}
r_{0}&:=r_{0}(X,Y)=\inf\{\|z\|: z\in X\cup Y, z\neq 0\}\label{eq:def_r0}\\
R&:=R(X,Y)=\sup\left\{\|z\|:z\in X\cup Y\right\}\label{eq:def_R}\\
K &:=K(R/r_{0},\gamma)=\lceil \log(R/r_{0}) / \log(1+\gamma) \rceil\label{eq:K_def}
\end{align}
Further for $i\in \mathbb{Z}$ define $r_{i}:=(1+\gamma)^{i-1}r_{0}$ and $S_{i}:=S_{i}(\gamma)=[r_{i}, r_{i+1})$ and the corresponding sets:
\begin{equation}
X_{i}:=\{x\in X\bigr|\|x\|\in S_{i}\}, \ i \in [K]
\end{equation}
For any point $x\in \R^{d}$ define $i(x):=\arg\min_{i\in \mathbb{Z}}\{\|x\|\in S_{i}\}$, and its \emph{norm-truncated version}:
\begin{equation}\label{eq:truncated}
\tilde{x}:=\tilde{x}_{\gamma} = \frac{x}{\|x\|} r_{i(x)} 
\end{equation}
For any point $x\neq 0$ let $\hat{x}:= \frac{x}{\|x\|}$. Note that   $\hat{x}= \frac{\tilde{x}}{r_{i(x)}}$ is also the normalized version of $\tilde{x}$. The motivation for partitioning the space in such annuli and projecting points on the inner boundary of each spherical annulus is that in doing so the ratio between the function $w(x,y)$ and $w(\tilde{x},\tilde{y})$ does not change too much.

\begin{lemma}\label{lem:ratio}For points $x,y\in \R^{d}$ such that $\|x\|,\|y\|\in [r_{0}, R]$  and $\gamma\in (0,1]$, let   $w(x,y)=p_{0}(\|x\|)e^{\phi(\langle x,y\rangle)}$ with $p_{0}$ being $(q, HR\gamma,\gamma)$-log-Lipshcitz and $\phi$ being $L$ Lipschitz. Then
\begin{equation}
e^{-\left(q+HR+3Lr_{i(x)}r_{i(y)}\right)\gamma}\leq \frac{w(\tilde{x}_{\gamma},\tilde{y}_{\gamma})}{w(x,y)}\leq e^{ \left(q+HR+3Lr_{i(x)}r_{i(y)}\right)\gamma}\label{eq:ratio}
\end{equation}
\end{lemma}
This suggests that if we pick $\gamma$ appropriately we can use the framework of Multi-resolution HBE to perform importance sampling for each annulus separately and bound  the variance of the overall estimator. 

\begin{theorem}\label{thm:euclidean}
For a set $X\subset \R^{d}$  and a set of possible queries $Y$ define $r_{0},R $ by \eqref{eq:def_r0},  \eqref{eq:def_R} respectively. For every convex function $\phi:[-R^{2},R^{2}]\to \R$ and a  $(q,HR\gamma,\gamma)$-log-Lipschitz function $p_{0}$, let $w(x,y)=p_{0}(\|x\|)e^{\phi(\langle x,y\rangle)}$. There exists  constants $\gamma^{*}\in (0,1]$,  $K^{*}$ and a distribution $\mathcal{D}^{*}$ such that for every $y\in Y$, the estimator $Z(y)\sim \mathcal{D}_{*}$ is unbiased $\E[Z(y)]=Z_{w}(y)=\mu$, $V$-bounded with  $V(\mu)=2e^{5/2}(8M_{\phi_{K^{*}K^{*}}}^{3}+1)\mu^{-1/2}$ and has complexity $O\left(d(K^{*})^{2}(L(\phi)R^{2})^{5/6}M_{\phi_{K^{*}K^{*}}}\right)$ where $M_{\phi_{K^{*}K^{*}}}=\exp(O\left(\left\{\log(L(\phi)(K^{*})^{2})L(\phi)(K^{*})^{2}\right\}^{2/3}\right))$.
\end{theorem}

Invoking Theorem \ref{thm:v-bounded} with the estimators given by Theorem \ref{thm:euclidean} results in a data structure to approximate $Z_{w}(y)$ for all $y\in Y$.

\subsection{Proof of Theorem \ref{thm:euclidean}}
\hphantom{p}
\paragraph{Step 1}
Our first concern is to pick a constant $\gamma\in (0,1]$ so that the partitioning scheme in subsection \ref{ssec:paritioning} is fully defined. The constant on one hand affects the space/time (complexity) it takes to evaluate our estimator and on the other hand the variance through the approximation $\langle x,y\rangle\approx \langle \tilde{x}_{\gamma},\tilde{y}_{\gamma}\rangle$. To simplify things we pick $\gamma$ so that the value of $w(x,y)$ changes at most by a factor of $e$ when projecting points on the inner boundary of the spherical annulus.
\begin{equation}
\gamma^{*}= 1 /\max\left\{1,q+HR+3LR^{2}\right\}
\end{equation}
For this choice by \eqref{eq:K_def} and $\log(1+x)\geq \frac{2x}{2+x}$ we get  $K^{*}= \lceil\frac{3}{2}\log(R/r_{0}) \max\{1, q+HR+3LR^{2}\}\rceil$. 
\paragraph{Step 2} For all pairs $i,j\in [K^{*}]$ we are going to construct an unbiased estimator for:
\begin{equation}
Z_{w}^{(ij)}(y) = \frac{\mathbb{I}\{\|y\|\in S_{j}\}}{nw_{\max}}\sum_{x\in X_{i}}p_{0}(\|x\|)e^{\phi(\langle x,y\rangle)} 
\end{equation}
It is easy to see that if  $\|y\|\in S_{j}$ then $Z_{w}(y)= \sum_{i\in [K^{*}]}Z_{w}^{(ij)}(y)$. For a  given pair $i,j\in [K^{*}]$, we  define a modified version of $\phi$. Let $\phi_{ij}:[-1,1]\to \R$ be the function given by $\phi_{ij}(\rho)=\phi(r_{i}r_{j}\rho)$ for all $\rho\in [-1,1]$ and set $\phi_{ij}^{*}=\sup\left\{\phi_{ij}(\rho)\bigr||\rho|\leq 1\right\}$. We are going to use these functions to perform ``importance sampling" in each spherical annulus $X_{i}$. To that end, we define for every pair $i,j\in [K^{*}]$:
\begin{align}
\mu_{ij}&:=\frac{1}{|X_{i}|e^{\phi_{ij}^{*}}}\sum_{x\in X_{i}}e^{\phi_{ij}(\langle \hat{x},\hat{y}\rangle)}\leq 1\label{eq:mu_ij}\\
A_{ij} &:= \frac{p_{0}(r_{i})|X_{i}|e^{\phi_{ij}^{*}}}{nw_{\max}} \leq 1\label{eq:A_ij}
\end{align}
Using these two quantities we can upper and lower bound the density $Z_{w}(y)$.

\begin{lemma} \label{lem:decomposition}
For any $y\in \R^{d}$ such that $\|y\|\in S_{j}$ we have for $\mu=Z_{w}(y)$ that
\begin{equation}
e^{-1}\cdot \sum_{i\in [K^{*}]}A_{ij}\mu_{ij}\leq \mu \leq e\cdot \sum_{i\in [K^{*}]}  A_{ij}\mu_{ij}
\end{equation} 
\end{lemma} 
\begin{proof} We only show the lower bound. Using Lemma \ref{lem:ratio} and the definition of $\gamma^{*}$ we get:
\begin{align}
\mu &=\frac{1}{nw_{\max}}\sum_{i\in [K^{*}]}\sum_{x\in X_{i}}w(x,y)\\
&\geq e^{-1}\frac{1}{nw_{\max}}\sum_{i\in [K^{*}]}\sum_{x\in X_{i}}w(\tilde{x},\tilde{y}) \\
&= e^{-1}\sum_{i\in [K^{*}]}\left(\frac{|X_{i}|p_{0}(r_{i})e^{\phi^{*}_{ij}}}{nw_{\max}}\right)\frac{1}{|X_{i}|e^{\phi_{ij}^{*}}}\sum_{x\in X_{i}}e^{\phi_{ij}(\langle \hat{x},\hat{y}\rangle)}
\end{align}
The upper bound follows similarly.
\end{proof}
Before constructing the estimators for $Z^{(ij)}_{w}(y)$, we relate  the Lipschitz constants of $\phi$ and $\phi_{ij}$.
\begin{proposition}[Rescaling]\label{prop:scaling}
 Given $\alpha>0$, and a convex function $\phi:[-a,a]\to \R$ with  constant $L$, the function $\phi(\alpha \rho)$ is convex and $\alpha L$-Lipschitz.
\end{proposition}
\begin{proof} Convexity is trivial, and $|\phi(\alpha \rho_{1})-\phi(\alpha \rho_{2})| \leq L|\alpha \rho_{1}-\alpha \rho_{2}|\leq L\alpha |\rho_{1}-\rho_{2}|$.
\end{proof}
Thus, under our assumption $L(\phi_{ij})\leq L r_{i}r_{j}$.

\paragraph{Step 3} For each $i,j\in [K^{*}]$, define $\hat{X}_{i}:=\{\hat{x}:x\in X_{i}\}$. Let $\{\mathcal{H}^{ij}_{t},p_{t}^{ij}\}_{t\in T_{ij}}$ be the hashing scheme resulting from invoking Theorem \ref{thm:log-convex} for $\phi_{ij}$, $\hat{X}_{i}$ and $\beta =1/2$. 
\begin{itemize}
\item \emph{Preprocessing:} for all $t\in [T_{ij}]$,  sample a hash function $h^{ij}_{t}\sim \mathcal{H}^{ij}_{t}$ and evaluate it  on $\hat{X}_{i}$ creating hash table $H^{ij}_{t}$ . Let $H^{ij}_{t}(z)\subseteq \hat{X}_{i}$ denote the hash bucket where $z\in \mathcal{S}^{d-1}$  maps to under $h^{ij}_{t}$.
\item \emph{Querying:} given a query $y$ ($\|y\|\in S_{j}$), for all $t\in [T_{ij}]$ let $\hat{X}^{ij}_{t}\sim H^{ij}_{t}(\hat{y})$ be a random element from $H^{ij}_{t}(\hat{y})$ or $\bot$ if $H^{ij}_{t}(\hat{y})=\emptyset$. Return $Z_{ij}(y)=\frac{1}{nw_{\max}}\sum_{t\in [T_{ij}]}\left\{ \frac{p^{ij}_{t}\left(\hat{X}^{ij}_{t},\hat{y}\right)}{W^{ij}\left(\hat{X}^{ij}_{t},\hat{y}\right)}|H^{ij}_{t}(\hat{y})|w(X_{t},y)\right\}$.
\end{itemize}
where $W^{ij}(x,y) = \sum_{t\in [T_{ij}]}(p^{ij}_{t}(x,y))^{2}$. For $\|y\|\in S_{j}$, we   denote this estimator as $Z_{ij}\sim \mathcal{D}_{ij}(y)$. The estimator is unbiased and has complexity $\mathcal{C}_{ij}$ bounded by $O(dL(\phi_{ij})^{5/6}M_{\phi_{ij}})$ where $M_{\phi_{ij}} = \exp(O\left(\left\{\log(L(\phi_{ij})) L(\phi_{ij})\right\}^{2/3}\right))$ and given explicitly below \eqref{eq:below} in the proof of Theorem \ref{thm:log-convex}. We next bound its variance. Towards that end, we define a different estimator:
\begin{align}
\tilde{Z}_{ij}&=\frac{1}{nw_{\max}}\sum_{t\in [T_{ij}]}\left\{ \frac{p^{ij}_{t}\left(\hat{X}^{ij}_{t},\hat{y}\right)}{W^{ij}\left(\hat{X}^{ij}_{t},\hat{y}\right)}|H^{ij}_{t}(\hat{y})|w(\tilde{X}_{t},\tilde{y})\right\}\\
&=\left(\frac{p_{0}(r_{i})|X_{i}|e^{\phi_{ij}^{*}}}{nw_{\max}}\right)\frac{1}{|X_{i}|e^{\phi_{ij}^{*}}}\sum_{t\in [T_{ij}]}\left\{ \frac{p^{ij}_{t}\left(\hat{X}^{ij}_{t},\hat{y}\right)}{W^{ij}\left(\hat{X}^{ij}_{t},\hat{y}\right)}|H^{ij}_{t}(\hat{y})|e^{\phi_{ij}(\langle \hat{x},\hat{y}\rangle)}\right\}
\end{align}  
For this estimator we get by \eqref{eq:mu_ij} and \eqref{eq:A_ij} that   $\E[\tilde{Z}_{ij}]=A_{ij}\mu_{ij}$. Furthermore, by our construction of $\{\mathcal{H}_{t}^{ij},p_{t}^{ij}\}_{t\in [T_{ij}]}$  and Theorem \ref{thm:scale-free} for $\beta=1/2$ it follows that:
\begin{equation}\label{eq:modified}
\E[\tilde{Z}_{ij}^{2}] \leq A_{ij}^{2}\cdot\left( 16M_{\phi_{ij}}^{3}+1\right)\mu_{ij}^{3/2}
\end{equation}
Finally, due to Lemma \ref{lem:ratio} we have that $Z_{ij}\leq e \tilde{Z}_{ij}$.

\paragraph{Step 4} We are now in position to define the final estimator and bound its variance. For $\|y\|\in S_{j}$ and $i\in [K^{*}]$,   let $Z_{ij}\sim \mathcal{D}_{ij}$  as before, and define:
\begin{equation}
Z_{j}(y) = \sum_{i\in [K^{*}]}Z_{ij}(y)
\end{equation}
The estimator is unbiased $\E[Z_{j}(y)]=Z_{w}(y)$ and the variance is bounded by
\begin{align}
\E[Z_{j}^{2}] &\leq  (\E[Z_{j}(y)])^{2} + \sum_{i\in [K^{*}]}\E[Z_{ij}^{2}]\\
&\leq \mu^{2} + e^{2}\sum_{i\in [K^{*}]} (16 M_{\phi_{ij}}^{3}+1) A_{ij}^{2}\mu_{ij}^{3/2}\\
&\leq \mu^{2} + e^{2}\sum_{i\in [K^{*}]} (16 M_{\phi_{ij}}^{3}+1) A_{ij}^{1/2}e^{3/2}(e^{-1}A_{ij}\mu_{ij})^{3/2}\\
&\leq \mu^{2} + e^{5/2} \max_{i\in [K^{*}]}\{16M_{\phi_{ij}}^{3}+1\} \left(e^{-1}\sum_{i\in [K^{*}]}A_{ij}\mu_{ij}\right)^{3/2}\\
&\leq \mu^{2} + e^{5/2}(16M_{\phi_{K^{*}K^{*}}}^{3}+1)\mu^{3/2}
\end{align}
where in the penultimate inequality we used $A_{ij}\leq 1$,  H\"{o}lder's inequality and super-additivity of $g(x):=x^{3/2}$. The final steps follows from Lemma \ref{lem:decomposition} and monotonicity of $g(x)$. This shows that our estimator is $V$-bounded with $V(\mu)=2e^{5/2}(8M_{\phi_{K^{*}K^{*}}}^{3}+1)\mu^{-1/2}$ and complexity $O\left(d(K^{*})^{2}(LR^{2})^{5/6}M_{\phi_{K^{*}K^{*}}}^{4}\right)$ with $M_{\phi_{K^{*}K^{*}}}=\exp(O\left(\left\{\log(L(K^{*})^{2})L(K^{*})^{2}\right\}^{2/3}\right))$.
\subsection{Proof of Lemma \ref{lem:ratio}}
We first show that for all $x_{1},x_{2}\in S_{i}(\gamma)$, $y_{1},y_{2}\in S_{j}(\gamma)$, and $\gamma\leq 1$ we have:
\begin{eqnarray}
\bigl|\|x_{1}\| -\|x_{2}\|\bigr| &\leq&  r_{i} \gamma\label{eq:fact1}\\
\bigl|\|x_{1}\|\|y_{1}\| - \|x_{2}\|\|y_{2}\|\bigr| &\leq& 3r_{i}r_{j} \gamma\label{eq:fact2}
\end{eqnarray}
To see the first part, assume without loss of generality  that  $\|x_{1}\|\geq \|x_{2}\|$ and $\|y_{1}\|\geq \|y_{2}\|$. We have for $z\in \{x,y\}$: $\|z_{1}\|  -\|z_{2}\| \leq (1+\gamma)^{i(z_{1})}r_{0} - (1+\gamma)^{i(z_{2})-1}r_{0} \leq (1+\gamma)^{i(z_{1})-1}r_{0}\gamma $. For the second part, we used the fact that $\gamma\leq 1$.
\begin{align}
\|y_{1}\|\|x_{1}\|-\|y_{2}\|\|x_{2}\| &\leq (1+\gamma)^{i(y_{1})+i(x_{1})}r_{0}^{2} - (1+\gamma)^{i(y_{2})+i(x_{2})-2}r_{0}^{2}\\
&\leq (1+\gamma)^{i(y_{1})+i(x_{1})-2}r_{0}^{2}\left((1+\gamma)^{2}-1 \right)\\
&=3r_{i}r_{j}\gamma
\end{align}  Using \eqref{eq:fact1},\eqref{eq:fact2} and the fact that $\langle x, y\rangle = \|x\|\|y\|\langle \hat{x},\hat{y}\rangle$ we get:
\begin{align}
\phi(\langle \tilde{x}, \tilde{y}\rangle)&\geq \phi(\langle x, y\rangle) - L(\phi) (\|x\|\|y\|-\|\tilde{x}\|\|\tilde{y}\|)|\langle \hat{x},\hat{y}\rangle|\geq \phi(\langle x, y\rangle) - 3L(\phi)r_{i(x)}r_{i(y)}\gamma\\ \phi(\langle \tilde{x}, \tilde{y}\rangle)&\leq \phi(\langle x, y\rangle) + L(\phi) (\|x\|\|y\|-\|\tilde{x}\|\|\tilde{y}\|)|\langle \hat{x},\hat{y}\rangle| \leq \phi(\langle x, y\rangle) + 3L(\phi)r_{i(x)}r_{i(y)}\gamma
\end{align}
Putting these two together and by the fact that $p_{0}$ is $(q,HR\gamma, \gamma)$-log-Lipschitz the statement follows.

\section{Importance Sampling for Vector Functions}\label{sec:gradients}

In this section, we show that for a class of unbiased estimators, that result from \emph{jointly sampling} a random weight function $U:X\cup\{\bot\}\to \R_{+}$ and a random point $Y\in X\cup \{\bot\}$ according to some \emph{balanced distribution}, the variance of an unbiased estimator for the sum of vectors  is bounded by that of the same distribution applied for the vector norms (Corollary \ref{cor:gradients}). The class of such estimators  include trivially classical importance sampling as well as   Hashing-Based-Estimators (Lemma \ref{lem:HBE}). Using this connection we will show how to estimate sum of gradients when the gradient norms are log-convex functions of the inner product.

\subsection{Randomly weighted estimators via Balanced distributions}
We start by defining a class of estimators that work by sampling a point $Y$ from $X\cup\{\perp\}$ and a,  possibly random and correlated with $Y$,  function $U:X\cup\{\perp\}\to \R_{+}$ with support possibly on a subset $S$ of $X$.
\begin{definition}[Balanced distribution]\label{def:reweighted}
 Given a finite set $S\subset X$, let $\mathcal{D}$ be a distribution of a pair of  random variables $(U,Y)\sim \mathcal{D}$ where $Y\in X\cup\{\perp\}$ and $U:X\cup \{\perp\}\to \R_{+}$. A distribution is called $S$-\emph{balanced} if $U(S^{c}\cup\{\perp\})=\{0\}$, and $\E[U(x)|Y=x]=\frac{1}{\P[Y=x]}\in (0,\infty)$  for all  $x\in S$.
\end{definition} 
Classical importance sampling schemes correspond to the case where  $U(x)=\frac{1}{\P[Y=x]}$ is a deterministic function of $x$. We show next that any such distribution, even with random $U$, can be used to create unbiased estimators for the sum of a function on $S$. 
\begin{lemma}[Moments]\label{lem:reweighted-moments}
 Let $S\subseteq X$, $f:X\cup \{\perp\}\to \R$ a bounded function, and $\mathcal{D}$ an $S$-balanced distribution. For $(U,Y)\sim \mathcal{D}$ it holds that
\begin{equation}
 \E[U(Y) f(Y)] = \sum_{x\in S} f(x) \qquad \text{and} \qquad 
\E[\{U(Y)f(Y)\}^{2}] = \sum_{x\in S}\frac{\E[U^{2}(x)|Y=x]}{\E[U(x)|Y=x]}f^{2}(x)
\end{equation}
\end{lemma}
\begin{proof} Using the law of total probability we have:
\begin{align}
\E[U(Y)f(Y)] &= \sum_{x\in S}\E[U(Y)f(Y)|Y=x]\P[Y=x]\\
&=\sum_{x\in S}f(x)\E[U(x)|Y=x]\P[Y=x]\\
&= \sum_{x\in S}f(x)
\end{align}
We proceed similarly:
\begin{align}
\E[\{U(Y)f(Y)\}^{2}] &= \sum_{x\in S}\E[\{U(Y)f(Y)\}^{2}|Y=x]\P[Y=x]\\
&=\sum_{x\in S}\E[U^{2}(x)|Y=x]\P[Y=x]f^{2}(x)\\
&= \sum_{x\in S}\frac{\E[U^{2}(x)|Y=x]}{\E[U(x)|Y=x]}f^{2}(x)
\end{align}
\end{proof}
Finally, we show that for vector functions the variance is controlled by the variance of the corresponding estimator for the sum of the gradient norms.
\begin{corollary}[Vectors to Norms]\label{cor:gradients} Let $g:X\cup \{\perp\}\to \R^{d}$ a bounded function, and $S\subseteq X$. For any $S$-balanced distribution $(U,Y)\sim \mathcal{D}$, we have $\E[U(Y)g(Y)]=\sum_{x\in S} g(x)$ and 
\begin{equation}
\E[\|U(Y)g(Y)\|^{2}] = \E[\left\{U(Y)\cdot \|g(Y)\|\right\}^{2}] = \sum_{x\in S}\frac{\E[U^{2}(x)|Y=x]}{\E[U(x)|Y=x]}\|g(x)\|^{2}
\end{equation}
\end{corollary}
\begin{proof}
The first equation follows   by applying Lemma \ref{lem:reweighted-moments} for $i\in [d]$, $g_{i}:X\to \R$ and linearity of expectation, while  the second part by applying the lemma for  $f(x)=\|g(x)\|$.
\end{proof}

\subsection{Hashing-Based-Estimators}
We next show that Hashing-Based-Estimators induce indeed balanced distributions for the support of the collision probability on $X$ for a given query $y$.
\begin{lemma}[HBE]\label{lem:HBE}
Given a set $X\subset \mathcal{X}$, and  a hashing scheme $\mathcal{H}$ with collision probabilities $p:\mathcal{X}\cup \{\perp\}\times \mathcal{X}\to [0,1]$, let $(h,g)\sim \mathcal{H}$. For any given $y\in \mathcal{X}$, let $Y\sim H_{X}(y)$ and $S(y):=\{x\in X|p(x,y)>0\}$, the distribution of $\left(\frac{|H(y)|}{p(Y,y)},Y\right)$ is $S(y)$-balanced.
\end{lemma}
\begin{proof} For all $x\in S(y)$,
\begin{align}
\E\left[\left.\frac{|H(y)|}{p(Y,y)}\right|Y=x\right] &=\frac{ \E\left[\frac{|H(y)|}{p(x,y)}\mathbb{I}[Y=x]\right]}{\P[Y=x]}\\
&= \frac{ \E\bigl[|H(y)|\mathbb{I}[Y=x]\mathbb{I}[x\in H(y)]\bigr]}{p(x,y)\P[Y=x]}\\
&= \frac{ \E\bigl[|H(y)|\mathbb{I}[Y=x]\bigr|x\in H(y)\bigr]p(x,y)}{p(x,y)\P[Y=x]}\\
&=\frac{1}{\P[Y=x]}\in (0,\infty)
\end{align}
\end{proof}
\subsection{Multi-resolution HBE}
To cover Multi-resolution HBE, or   their    Multi-scale extension described in Section \ref{sec:reduction}, we show that adding together randomly weighted estimators, resulting from balanced distributions that are pairwise independent, produces the results we expect. 
\begin{corollary} \label{cor:gradient}
Given $X\subset \mathcal{X}$, $y\in \mathcal{X}$, let $(U_{t},Y_{t})\sim \mathcal{D}_{t}(y)$ for $t\in [T]$ being pairwise independent and $D_{t}(y)$ t being $S_{t}(y)$-balanced. Let $T(x,y) = \{t\in [T]|x\in S_{t}(y)\}$. For  a collection of  bounded  functions $\{f_{t}:\mathcal{X}\cup \{\perp\}\to \R^{d}\}_{t\in [T]}$, we have:
\begin{equation}
\E[\sum_{t\in [T]} U_{t}(Y_{t}) f_{t}(Y_{t})] = \sum_{x\in X}\sum_{t\in T(x,y)}f_{t}(x)
\end{equation}
and $\E[\|\sum_{t\in [T]} U_{t}(Y_{t}) f_{t}(Y_{t})\|^{2}]\leq\sum_{t\in [T]} \E[\{U_{t}(Y_{t})\| f_{t}(Y_{t})\|\}^{2}] + \left(\E[\sum_{t\in [T]}U_{t}(Y_{t}) \|f_{t}(Y_{t})\|]\right)^{2}$.
\end{corollary}
\begin{proof} The first part follows easily due to linearity and Lemma \ref{lem:reweighted-moments}, while the second one follows from triangle inequality.
\end{proof}

This shows that if Multi-resolution HBE has small variance in estimating the sum of the vector norms, it  can be used to estimate the sum of the vectors with the    same variance up to constants.

\begin{corollary}
Let $g:\mathcal{S}^{d-1}\times \mathcal{S}^{d-1}\to \R^{m}$ be  a vector function  such that $\|g(x,y)\|_{2}=e^{\phi(\langle x, y\rangle)}$ for some convex function $\phi$. Given $\epsilon,\tau\in (0,1)$, there exists an explicit constant $M_{\phi}$ and a data structure using space $O\left(dL(\phi)^{5/6}M_{\phi}^{3}\frac{1}{\epsilon^{2}}\frac{1}{\sqrt{\tau}} \cdot n\right)$ and query time $O(dL(\phi)^{5/6}M_{\phi}^{4}\frac{1}{\epsilon^{2}}\frac{1}{\sqrt{\mu}})$ that for any $y\in \mathcal{S}^{d-1}$ with constant probability can either produce a   a vector $G$ such that:
\begin{equation}
\left\|G- \frac{1}{ne^{\phi_{\max}}}\sum_{x\in X}g(x,y)\right\|_{2}\leq \epsilon \mu
\end{equation}
if $\mu:=\frac{1}{ne^{\phi_{\max}}}\sum_{x\in X}\|g(x,y)\|_{2} \geq \tau$ or assert  that $\mu < \tau$.
\end{corollary}
\begin{proof} We first call Theorem \ref{thm:main} to construct an Multi-resolution HBE for the problem of approximating $Z_{\phi}(y)$, where $\phi=\log(\|g\|_{2})$. By Corollary \ref{cor:gradient}, this shows that we can turn our MR-HBE estimator to an unbiased  estimator for $\sum_{x\in X}g(x,y)$ and that the variance is bounded by that of estimating $Z_{\phi}(y)$. 
\end{proof}

\section{Lower bound under SETH or OVC}\label{sec:lower}
 
\begin{conjecture}[Strong Exponential Time Hypothesis (SETH)\cite{impagliazzo2001complexity}]
For any $\epsilon>0$, there exists $k=k(\epsilon)$ such that $k$-SAT on $n$ variables cannot be solved in time $O(2^{(1-\epsilon)n})$.
\end{conjecture}
A conjecture that is implied by SETH~\cite{williams2005new,williams2014finding}, concerns the complexity of finding a pair of orthogonal vectors amongst two set of  binary vectors.
\begin{conjecture}[Orthogonal Vectors Conjecture (OVC)]
For every $\delta>0$ there exists $c=c(\delta)$ such that given two sets $A,B\subset \{0,1\}^{m}$ of cardinality $N$, where $m=c\log N$, deciding if there is a pair $(a,b)\in A\times B$ such that $a\top b=0$ cannot be solved in time $O(N^{2-\delta})$.
\end{conjecture}
These popular conjectures have been the base of a flurry of quadratic hardness results in the past years.  The basis of our hardness result is the following recent theorem by Aviad Rubinstein~\cite{rubinstein2018hardness}.  Let $d^{2}(A,B):=\min\limits_{a\in A}\min\limits_{b\in B}\{\|a-b\|^{2}_{2}\}$ be the minimum squared distance between  $A,B\subset \R^{d}$.

\begin{theorem}[Theorem 4.1\cite{rubinstein2018hardness}]\label{thm:aviad}
Unless SETH and OVC are false, the following holds:  for every $\delta>0$ and $\epsilon\in (0,e^{-1})$ there exist  constants $c(\delta)>0$, $T(\epsilon)=O(\frac{\log \frac{1}{\epsilon}}{\log \log \frac{1}{\epsilon}})$ and $T'=2^{O(T\log T)}=O(\frac{1}{\epsilon})$ such that given two sets $A,B\subset \{0,1\}^{d}$ of $N$ vectors with
\begin{itemize}
\item Dimension: $d \geq  2 mT'$, with $m = c(\delta)\log N$
\item Sparisty: for all $x\in A\cup B$, $\|x\|_{2}^{2}= mT'$
\end{itemize}
there is no algorithm that decides whether $d^{2} (A,B) \left\{ \begin{matrix} = m (T'-1) \\
\textrm{or}\\
\geq mT'\qquad
\end{matrix} \right.$
   in time $N^{2- O\left(\delta + c(\delta)\frac{\log^{2}\log \frac{1}{\epsilon}}{\log \frac{1}{\epsilon}}\right)}$.
\end{theorem}

Our proof will proceed by translating hardness for the problem of Approximate Bi-chromatic Closest pair to our setting. This connection was first established  in  \cite{backurs2017fine} to obtain quadratic hardness results for Kernel Methods and Neural Networks.

\subsection{Proof of Theorem \ref{thm:lower}}

\begin{proof} 
The proof proceeds by showing how to reduce an instance $(A,B)$ of the approximate Bi-chromatic closest pair   in Theorem \ref{thm:aviad} to an instance $(X,Y) \subset \mathcal{S}^{d-1}\times \mathcal{S}^{d-1}$ of producing a $\alpha$ approximation to: $\frac{1}{N^{2}}\sum_{x\in X}\sum_{y\in Y} e^{L \cdot (\langle x, y\rangle-1)}$.
\paragraph{Setting $\epsilon = e^{-e^{\delta/c(\delta)}}\in (0,e^{-1})$ in Theorem \ref{thm:aviad}}
We start by finding a constant $\epsilon \in (0,e^{-1})$ such that:
\begin{align}
&\qquad c(\delta)\frac{\log^{2}\log \frac{1}{\epsilon}}{\log \frac{1}{\epsilon}} \leq \delta\\
&\Leftrightarrow \log^{2}\log\frac{1}{\epsilon} \leq \left(\frac{\delta}{c(\delta)}\right) \log\frac{1}{\epsilon}\\
& \Leftrightarrow \zeta^{2} \leq  \left(\frac{\delta}{c(\delta)}\right) e^{\zeta}
\end{align}
where $\zeta = \log \log\frac{1}{\epsilon} >0$. Setting $\zeta = \frac{\delta}{c(\delta)}>0$ we get  $e^{\zeta} \geq 1+ \zeta \geq \zeta>0 $.  For this choice we have:
\begin{align}
\epsilon = e^{-e^{\zeta}} < e^{-1} \Leftrightarrow e^{\zeta}>1
\end{align}
Hence, we may pick $\epsilon = e^{-e^{\frac{\delta}{c(\delta)}}}$ for which $\tilde{T}(\delta) =  O(\frac{e^{\delta/c(\delta)}}{\delta/c(\delta)})$  and $\tilde{T}'(\delta) = O(e^{e^{\delta/c(\delta)}})$. Theorem \ref{thm:aviad} then shows that 
there is no $N^{2-O(\delta)}$ algorithm to decide between: $d^{2}(A, B) \left\{ \begin{matrix} = m (\tilde{T}'(\delta)-1) \\
\textrm{or}\\
\geq m\tilde{T}'(\delta)\qquad
\end{matrix} \right.$. 

\paragraph{Translating distance bounds to Density bounds for Gaussian Kernel}
We next show that distinguishing between the two cases for $d^{2}(A,B)$ distinguishes between two values for the average of  the Gaussian kernel between points in the two datasets.
In the case where $d^{2}(A,B)\geq m \tilde{T}^{'}(\delta)$, we have that:
\begin{align}
\frac{1}{N^{2}}\sum_{a \in A}\sum_{b\in B}e^{-\beta \|a-b\|^{2}} \leq e^{-\beta m \tilde{T}^{'}(\delta)}
\end{align}
In the other case, where $d^{2}(A,B)=m(\tilde{T}'(\delta)-1)$ we get:
\begin{align}
\frac{1}{N^{2}}\sum_{a \in A}\sum_{b\in B}e^{-\beta \|a-b\|^{2}} \geq \frac{1}{N^{2}} e^{-\beta d^{2}(A,B)} = e^{-\beta m \tilde{T}^{'}(\delta)}\cdot e^{-2\log N + \beta m}
\end{align}
So as long as $e^{-2\log N+\beta m} > \alpha   \Leftrightarrow \beta > \frac{2\log N + \log\alpha}{m}$ any algorithm that can produce a $\alpha$-approximation to $\frac{1}{N^{2}}\sum_{a \in A}\sum_{a\in B}e^{-\beta\|a-b\|^{2}}$ distinguishes between the two cases as such it cannot run in time $N^{2-O(\delta)}$. 
\paragraph{Gaussian Kernel to Log-convex (linear) and Bound on Lipschitz Constant}
To complete the proof we observe that:
\[
\beta \|a-b\|^{2}= -\beta 2m\tilde{T}' (\langle\frac{a}{\sqrt{m\tilde{T}'}},\frac{b}{\sqrt{m\tilde{T}'}}-1\rangle) = L(\langle\frac{a}{\sqrt{m\tilde{T}'}},\frac{b}{\sqrt{m\tilde{T}'}}\rangle-1)
\]  with $L:=2\beta m\tilde{T}'$. Setting $Y:=\{a /\sqrt{m\tilde{T}'}: a\in  A \}$ and $X:=\{b/\sqrt{m\tilde{T}'}: b\in  B \}$ we have that:
\[
e^{-\beta \|a-b\|^{2}} = e^{L (\langle y,x\rangle -1)}
\]
and $X,Y\subset \mathcal{S}^{d-1}$. Hence, substituting the lower bound on $\beta$ we get that for:
\[
L > 2\tilde{T}'(\delta) (2\log N + \log \alpha) = \left\{C(\delta) \left(1 + \frac{\log  \alpha }{2\log N}\right)\right\}\cdot \log N
\]
where $C(\delta) = O\left(e^{e^{\frac{\delta}{c(\delta)}}}\right)$
there is no algorithm that approximates the sum in time less than $N^{2-O(\delta)}$.
\end{proof}
\section{Remaining Proofs}\label{sec:proofs}
This section contains proofs of lemmas and theorems stated in the main paper as well as various auxiliary results. 

\subsection{Proof of Corollary \ref{cor:examples}}
Under the condition $r\leq \frac{1}{2}\sqrt{\log n}$ we have that the Lipschitz constants of the first four functions in Table \ref{tbl:examples} are bounded by $L(\phi)\leq 2r^{2} \leq \frac{1}{2}\log n$. This is also true  for the last function under the condition $0\leq k \leq \frac{c-1}{2}\log n$.  The result follows from $\mu \in [e^{-2L(\phi)},1]\subseteq [\frac{1}{n},1]$.
\subsection{Moments of Multi-resolution HBE}

\begin{proof}[Proof of Lemma \ref{lem:moments}]
We start by computing the first moment:
\begin{align}
\E[Z_{T}(y)] &= \frac{1}{|X|}\sum_{t\in [T]} \E\left[\frac{w_{t}( X_{t},y)}{p_{t}(X_{t},y)}|H_{t}(y)|\right]\\
&=\frac{1}{|X|}   \sum_{x\in X}\sum_{t\in T(x,y)}w_{t}( x, y )\\
&=  \frac{1}{|X|}   \sum_{x\in X}w(x, y)
\end{align}
The second moment is given by
\begin{align}
\E[Z_{T}^{2}] & = \frac{1}{|X|^{2}}\sum_{t\in [T]} \sum_{t^{'}\in [T]} \E\left[ \frac{w_{t}( X_{t}, y)}{p_{t}(X_{t},x)}|H_{t}(y)|\frac{w_{t^{'}}(  X_{t^{'}},x)}{p_{t^{'}}(X_{t^{'}},x)}|H_{t^{'}}(y)|\right]\\
& \leq\frac{1}{|X|^{2}}  \sum_{t\in [T]} \E\left[\frac{w_{t}^{2}(X_{t},y)}{p^{2}_{t}(X_{t},y)}|H_{t}(y)|^{2}\right] + \mu^{2}\\
& = \frac{1}{|X|^{2}}\sum_{x\in X} \sum_{t\in T(x,y)}\frac{w_{t}^{2}(x, y)}{p_{t}(x,y)}\E\left[|H_{t}(y)||x\in H_{t}(y)\right] + \mu^{2}\\
&\leq  \frac{1}{|X|^{2}}\sum_{x\in X} \sum_{t\in T(x,y)}\frac{w_{t}^{2}(x, y)}{p_{t}(x,y)}\sum_{z\in X}\frac{\min\{p_{t}(z,y),p_{t}(x,y)\}}{p_{t}(x,y)} + \mu^{2}
\end{align}
\end{proof}

\subsection{Distance Sensitive Hashing on the unit sphere}
To analyze the collision probability of the DSH scheme we closely follow  the proof of Aumuller et al.~\cite{aumuller2018distance} with the difference that we use Proposition \ref{prop:tails_two} to bound bi-variate Gaussian integrals. 
\begin{proposition}[Proposition 3~\cite{szarek1999nonsymmetric}]\label{prop:szarek}
 Let $X_{1}\sim N(0,1)$ and $t>0$
\begin{equation}
\frac{1}{\sqrt{2\pi}}\frac{1}{t+1}e^{-\frac{t^{2}}{2}} \leq \P[X_{1}\geq t] \leq \frac{1}{\sqrt{2\pi}}\frac{1}{t}e^{-\frac{t^{2}}{2}}\label{eq:tails_one}
\end{equation}
\end{proposition}
\begin{proposition}[Propositions 3.1 \& 3.2~\cite{hashorva2003multivariate}]\label{prop:tails_two}
 Let $(X_{1},X_{2})\sim \mathcal{N}(0,\begin{bmatrix}
1&\rho\\
\rho & 1
\end{bmatrix})$ be two $\rho$-correlated standard normal random variables. For all $\rho <1$ and $t>0$:
\begin{align}
 \P[X_{1}>t  \wedge X_{2}>t]&\geq \frac{4}{(1+\sqrt{1+4\frac{(1+\rho)^{2}}{\min(1-\rho,1+\rho)}})^{2}}\frac{\min(1-\rho,1+\rho)}{(1+\rho)^{2}}\frac{1+|\rho|}{2\pi\sqrt{1-\rho^{2}}} e^{-\frac{2}{1+\rho}\frac{t^{2}}{2}}\\   
\P[X_{1}>t  \wedge X_{2}>t] &\leq\frac{(1+\rho)^{\frac{3}{2}}}{2\pi \sqrt{1-\rho}} e^{-\frac{2}{1+\rho}\frac{t^{2}}{2}}
 \end{align}
\end{proposition}
We first simplify the sub-exponential terms appearing on the above inequalities using our assumption that $|\rho| < 1-\delta$.  Since the function $\frac{(1+\rho)^{\frac{3}{2}}}{2\pi\sqrt{1-\rho}}$ is increasing in $\rho$ we get $\frac{(1+\rho)^{\frac{3}{2}}}{2\pi\sqrt{1-\rho}} \leq \frac{\sqrt{2}}{\pi \sqrt{\delta}}$. Additionally, we have that $\frac{\min(1-\rho, 1+\rho)}{(1+\rho)^{2}} \geq \frac{\delta}{4}$  and  $(a+b)^{2}\leq 2(a^{2}+b^{2})$ for all $a,b\in \R$. Using the above bounds we get:
\[
\frac{4}{(1+\sqrt{1+4\frac{(1+\rho)^{2}}{\min(1-\rho,1+\rho)}})^{2}}\frac{\min(1-\rho,1+\rho)}{(1+\rho)^{2}}\frac{1+|\rho|}{2\pi\sqrt{1-\rho^{2}}} \geq \frac{2}{2+\frac{16}{\delta}} \frac{\delta}{4}\frac{1}{2\pi} \geq \frac{\delta^{2}}{8+\delta}\frac{1}{8\pi}
\]
We are now in a position ot bound the collision probability.
\begin{proof}[Proof of Lemma \ref{lem:collision_basic}]
The collision probability can be written as:
\begin{align}
\P[h(x)=g(y)] = \P[h(x)\leq m \wedge g(y)\leq m] \frac{\P[\langle x, g\rangle\geq t\wedge \langle y, g\rangle\geq t ]}{\P[\langle x, g\rangle\geq t \vee \langle y, g\rangle\geq t]}\label{eq:def_prob}
\end{align}
We are going to obtain upper and lower bounds for both terms. We start first with the second term. An easy calculation shows that the vector $(X_{1},X_{2}):=(\langle x, g\rangle, \langle y, g\rangle)\sim \mathcal{N}(0,\begin{bmatrix}
1&\rho\\
\rho & 1
\end{bmatrix})$ follows a bivariate normal distribution with unit variances and correlation $\rho=\langle x, y\rangle$. Hence,   $\P[\langle x, g\geq t\rangle] =  \P[\langle y, g\geq t\rangle] = \P[X_{1}\geq t]$ and  $\P[\langle x, g\rangle\geq t\wedge \langle y, g\rangle\geq t ] = \P[X_{1}\geq t\wedge X_{2}\geq t]$.  Using monotonicity and union bound we get that:
\begin{equation}
\frac{1}{2}\frac{\P[X_{1}\geq t \wedge X_{2}\geq t]}{\P[X_{1}\geq t]}\leq \frac{\P[\langle x, g\rangle\geq t\wedge \langle y, g\rangle\geq t ]}{\P[\langle x, g\rangle\geq t \vee \langle y, g\rangle\geq t]} \leq \frac{\P[X_{1}\geq t \wedge X_{2}\geq t]}{\P[X_{1}\geq t]}\label{eq:two_intermediate}
\end{equation}
Using \eqref{eq:two_intermediate} and the  estimates from Propositions \ref{prop:szarek}, \ref{prop:tails_two}
\begin{align}
\frac{\sqrt{2}\delta^{2}}{148\sqrt{\pi}} e^{-\frac{1-\rho}{1+\rho}\frac{t^{2}}{2}}\leq \frac{\P[\langle x, g\rangle\geq t\wedge \langle y, g\rangle\geq t ]}{\P[\langle x, g\rangle\geq t \vee \langle y, g\rangle\geq t]}  \leq \frac{2}{\sqrt{\pi}\sqrt{\delta}} e^{-\frac{1-\rho}{1+\rho}\frac{t^{2}}{2}}\label{eq:two_final}
\end{align}
Next, we bound the remaining term as
\begin{align}
  \P[h(x)\leq m \wedge g(y)\leq m] &\geq 1 - \P[h(x) > m \wedge g(y)> m]\nonumber\\
&\geq 1- 2(1-\P[X_{1}\geq t])^{m}\nonumber\\
& \geq 1- 2e^{-\P[X_{1}\geq t]m}\nonumber\\
&\geq 1 - \zeta
\end{align}
where in the last step we used the definition of $m(t,\zeta)$ and the lower bound from \eqref{eq:tails_one}. Using the last inequality along with \eqref{eq:two_final} and \eqref{eq:def_prob}, we arrive at:
\begin{equation}
\frac{\sqrt{2}(1-\zeta)\delta^{2}}{148\sqrt{\pi}} e^{-\frac{1-\rho}{1+\rho}\frac{t^{2}}{2}}\leq p_{+}(\rho) \leq \frac{2}{\sqrt{\pi}\sqrt{\delta}} e^{-\frac{1-\rho}{1+\rho}\frac{t^{2}}{2}}
\end{equation} 
Next, we treat the case where $1-\delta < \rho \leq 1$, let $Z_{1},Z_{2}$ be standard normal random variables then:
\begin{align}
1\geq \frac{\P[Z_{1}\geq t \wedge \rho Z_{1}+\sqrt{1-\rho^{2}}Z_{2}\geq t]}{\P[Z_{1}\geq t \vee \rho Z_{1}+\sqrt{1-\rho^{2}}Z_{2}\geq t]} \geq \frac{1}{2}\P[Z_{2}\geq \sqrt{\frac{1-\rho}{1+\rho}} t] 
\end{align}
\begin{align}
\P[Z_{2}\geq \sqrt{\frac{1-\rho}{1+\rho}} t] &\geq \frac{1}{\sqrt{2\pi}}\frac{\sqrt{1+\rho}}{\sqrt{1-\rho}+\sqrt{1+\rho}} e^{-\frac{1-\rho}{1+\rho} \frac{t^{2}}{2}}\geq \frac{1}{\sqrt{2\pi}}\frac{1}{1+\sqrt{2}} e^{-\frac{\delta}{2-\delta} \frac{t^{2}}{2}} 
\end{align}
Lastly, we show an upper bound on $p_{+}(\rho)$ for $-1\leq \rho \leq -1+\delta$, we have that:
\begin{align}
p_{+}(\rho) &\leq \frac{\P[Z_{1}\geq t\wedge  \rho Z_{1}+\sqrt{1-\rho^{2}}Z_{2}\geq t]}{\P[Z_{1}\geq t \vee \rho Z_{1}+\sqrt{1-\rho^{2}}Z_{2}\geq t]} \\
&\leq  \frac{1}{\P[Z_{1}\geq t]}\int_{t}^{\infty}\P[Z_{2}\geq  \frac{t-\rho u }{\sqrt{1-\rho^{2}}}] \frac{1}{\sqrt{2\pi}}e^{-\frac{u^{2}}{2}}du\\
&\leq 
\frac{1}{\P[Z_{1}\geq t]}\int_{t}^{\infty}\P[Z_{2}\geq  \frac{t-(-1+\delta) u }{\sqrt{1(-1+\delta)^{2}}}] \frac{1}{\sqrt{2\pi}}e^{-\frac{u^{2}}{2}}du\\
& \leq \frac{2}{\sqrt{\pi}\sqrt{\delta}} e^{-\frac{2}{\delta}\frac{t^{2}}{2}}
\end{align}
This concludes the proof.
\end{proof}
\subsection{Idealized Hashing}
We consider the idealized hashing probability $h_{\gamma,t}(\rho) =  -\left(\frac{1-\rho}{1+\rho}+\gamma^{2}\frac{1+\rho}{1-\rho}\right)\frac{t^{2}}{2}$. Its first and second derivatives are given by:
\begin{eqnarray}
h^{'}_{\gamma,t}(\rho) &=& \left(\frac{1}{(1+\rho)^{2}} - \gamma^{2} \frac{1}{(1-\rho)^{2}} \right)t^{2}\label{eq:hash_derivative}\\
h^{''}_{\gamma,t}(\rho) &=& -2 \left(\frac{1}{(1+\rho)^{3}}-\gamma^{2}\frac{1}{(1-\rho)^{3}} \right)t^{2}\label{eq:hash_second}
\end{eqnarray}
 
\vskip 0.2in
\begin{proof}[Proof of Proposition \ref{prop:derivatives}]  Using \eqref{eq:hash_derivative}, we see that the derivative becomes zero only at $\rho^{*}(\gamma) = \frac{1-\gamma}{1+\gamma}$ and that the second derivative becomes zero at $\rho^{**}(\gamma) = \frac{1-\gamma^{\frac{2}{3}}}{1+\gamma^{\frac{2}{3}}}$. Let $g(x)=\frac{1-x}{1+x}$, the function $h_{\gamma,t}$ is  concave for all $\rho\geq \rho^{**}(\gamma)=g(\gamma^{\frac{2}{3}})$. Since $g$ is decreasing for all $\rho\geq -1$, we have:
\begin{align}
\gamma \leq 1 \Rightarrow \gamma \leq \gamma^{\frac{2}{3}} \Rightarrow g(\gamma) \geq g(\gamma^{\frac{2}{3}}) \Leftrightarrow \rho^{*}(\gamma) \geq \rho^{**}(\gamma)\\
\gamma \geq 1 \Rightarrow \gamma \geq \gamma^{\frac{2}{3}} \Rightarrow g(\gamma) \leq g(\gamma^{\frac{2}{3}}) \Leftrightarrow \rho^{*}(\gamma) \leq \rho^{**}(\gamma)
\end{align}
\end{proof}
\begin{proof}[Proof of Proposition \ref{prop:infimum}]
We only show the case where $\phi$ is non-decreasing the other case follows similarly. We have that $g(\rho)\leq g(\rho_{*})$ for all $\rho\in [-1,1]$. By concavity, we know that:
\[
g(\rho)\leq g(\rho_{0})+g^{'}(\rho_{0}) (\rho-\rho_{0}), \ \forall \rho \in [-1, \rho^{*}]
\]
Therefore, we have that for all $\rho\in [1,\rho_{*}]$
\begin{align}
\phi(\rho)-g(\rho) &\geq \phi(\rho) - g(\rho_{0}) - g^{'}(\rho_{0}) (\rho-\rho_{0}) \\
&\geq  \phi^{'}(\rho_{0})-g(\rho_{0})+[\phi^{'}(\rho_{0})-g^{'}(\rho_{0})](\rho-\rho_{0}) \\
&= \phi(\rho_{0})-g(\rho_{0}) 
\end{align}
Finally, for $\rho\in [\rho^{*},1]$ we have by monotonicity $\phi(\rho) - g(\rho) \geq \phi(\rho_{*})  - g(\rho_{*}) \geq \phi(\rho_{0})-g(\rho_{0})$.
 \end{proof} 
%
\begin{proof}[Proof of Corollary \ref{col:exponent_bounds}]
Using the fact that $ a+b\leq 2 \max\{a,b\}$ and estimates from Lemma \ref{lem:positive}, we get that
\begin{align}
t_{0}^{2} &= -\frac{1}{2}\frac{1+\rho_{0}}{1-\rho_{0}} (2\phi(\rho_{0})-(1-\rho_{0}^{2})\phi^{'}(\rho_{0})) \leq - 2\frac{1+\rho_{0}}{1-\rho_{0}}\phi(\rho_{0}) \\
\gamma^{2}_{0}t_{0}^{2} &= -\frac{1}{2}\frac{1-\rho_{0}}{1+\rho_{0}} (2\phi(\rho_{0})+(1-\rho_{0}^{2})\phi^{'}(\rho_{0}))
\leq -2\frac{1-\rho_{0}}{1+\rho_{0}}\phi(\rho_{0}) 
\end{align}
When $\phi^{'}(\rho_{0})\geq 0$, we get by \eqref{eq:der_increasing} that:
\begin{align}
t_{0}^{2} & \geq -\frac{1+\rho_{0}}{1-\rho_{0}}\phi(\rho_{0})\\
\gamma_{0}^{2}t_{0}^{2}&\geq -\frac{1}{2}\frac{1-\rho_{0}}{1+\rho_{0}}2 \phi(\rho_{0})\frac{1-\rho_{0}}{2} \geq -\frac{(1-\rho_{0})^{2}}{2(1+\rho_{0})}\phi(\rho_{0})
\end{align}
Similarly, when $\phi^{'}(\rho_{0})\leq 0$, we get by \eqref{eq:der_decreasing}:
\begin{align}
t_{0}^{2} & \geq  -\frac{1}{2}\frac{1+\rho_{0}}{1-\rho_{0}}2 \phi(\rho_{0})\frac{1+\rho_{0}}{2} \geq -\frac{(1+\rho_{0})^{2}}{2(1-\rho_{0})}\phi(\rho_{0}) \\
\gamma_{0}^{2}t_{0}^{2}&\geq -\frac{1-\rho_{0}}{1+\rho_{0}}\phi(\rho_{0})
\end{align}
Using again $\max\{a,b\}\geq \frac{a+b}{2}$, we get in both cases that $\max\{\gamma_{0}^{2}t_{0}^{2},t_{0}^{2}\}\geq  - \frac{1+\rho_{0}^{2}}{1-\rho_{0}^{2}}\phi(\rho_{0})$.
\end{proof}
\subsection{Approximation}
\noindent
{\bf Proof of Lemma \ref{lem:linear}} The idea is to select a set of points $\rho_{1},\ldots, \rho_{T}$ and break $[\rho_{-},\rho_{+}]$ in intervals $\rho_{i}\leq \rho\leq \rho_{i}+\Delta(\rho_{i})$ of length $\Delta(\rho_{i})$ such that within each interval $\ell(\rho)$ is well approximated by $h_{\rho_{i}}(\rho)$. For $\rho\geq \rho_{0}$ using the Taylor Remainder theorem, there exists $\xi=\xi(\rho,\rho_{0})\in [\rho_{0},\rho]$ such that
\begin{align}
\ell(\rho)-h_{\rho_{0}}(\rho) &= [\ell(\rho_{0})- h(\rho_{0})]+[(\ell^{'}(\rho_{0})-h^{'}(\rho_{0}))(\rho-\rho_{0})] -\frac{1}{2}h^{''}_{\rho_{0}}(\xi(\rho,\rho_{0}))(\rho-\rho_{0})^{2}\nonumber\\
& = -\frac{1}{2}h^{''}_{\rho_{0}}(\xi(\rho,\rho_{0}))(\rho-\rho_{0})^{2} \geq 0\label{eq:remainder}
\end{align}
Where the inequality follows by concavity of $h$. To obtain an upper bound, we need an absolute bound on the second derivative. Using \eqref{eq:hash_second}, we get that
\begin{align}
|h^{''}_{\gamma_{0},t_{0}}| &\leq 2\max\left\{\frac{1}{(1+\rho)^{3}}t_{0}^{2},\gamma^{2}_{0}t_{0}^{2}\frac{1}{(1-\rho)^{3}}\right\}
\end{align}
Substituting the upper bounds from Corollary \ref{col:exponent_bounds} in turn gives
\begin{align}
|h^{''}_{\gamma_{0},t_{0}}| \leq 8 \max\left\{\frac{1+\rho_{0}}{1-\rho_{0}}\frac{1}{(1+\rho)^{3}},\frac{1-\rho_{0}}{1+\rho_{0}}\frac{1}{(1-\rho)^{3}} \right\}R(\ell)
\end{align}
For $\rho_{0}\leq \rho\leq \rho_{0}+\Delta(\rho_{0})\leq 0$ we have $|h^{''}_{\rho_{0}}| \leq \frac{16}{(1-|\rho_{0}|)^{2}}R(\ell)$. Setting $\Delta(\rho_{0}) = \sqrt{\frac{\epsilon}{8 R(\ell)}}(1-|\rho_{0}|)$, gives
\begin{equation}
\ell(\rho)-h_{\rho_{0}}(\rho) \leq  \frac{1}{2}|h^{''}_{\rho_{0}}|\Delta^{2}(\rho_{0})\leq \frac{8}{(1-|\rho_{0}|)^{2}}|\ell_{\min}|\Delta^{2}(\rho_{0}) \leq \epsilon
\end{equation}
Hence, we have the following inductive definition of points $\rho_{i}$:
\begin{align}
1+\rho_{i} &= 1+\rho_{i-1}+\Delta(\rho_{i-1})\\
 &= (1 + \rho_{i-1})+ \sqrt{\frac{\epsilon}{8R(\ell)}}(1+\rho_{i-1})\\
 & = (1+\sqrt{\frac{\epsilon}{8R(\ell)}}) (1+\rho_{i-1})
\end{align}
multiplying both sides with $\sqrt{\frac{\epsilon}{8R(\ell)}}$ gives us the updates for $\Delta(\rho_{i})$. We are now in a position to write an explicit expression for $\rho_{i}$:
\begin{align}
\rho_{i}  &= \rho_{-}+ \sum_{j=1}^{i}\Delta(\rho_{j-1})\\
&=\rho_{-} + \sum_{j=1}^{i}\left(1+\sqrt{\frac{\epsilon}{8R(\ell)}}\right)^{j-1}\sqrt{\frac{\epsilon}{8R(\ell)}} (1-|\rho_{-}|)\\
&= \rho_{-} + \left[\left(1+\sqrt{\frac{\epsilon}{8R(\ell)}} \right)^{i}-1 \right](1-|\rho_{-}|)
\end{align}
for $i=0,\ldots, T$ with   $T= \lfloor \frac{\log(\frac{1-|\rho_{+}|}{1-|\rho_{-}|})}{\log(1+\sqrt{\frac{\epsilon}{8R(\ell)}})}\rfloor $. The floor function is justified by the fact that if $\rho_{T}<\rho_{+}$ then $\rho_{T}+\Delta_{T}>\rho_{+}$ and as such $\phi$ is well approximated between $[\rho_{T},\rho_{+}]$ by $h_{\rho_{T}}$. The lemma follows by setting $i(\rho):=\min\{j\in\{0,\ldots, T\}|\rho_{i}\leq \rho\}$.
\hfill $\blacksquare$
\subsection{Scale-free Multi-resolution HBE}


\begin{proof}[Proof of Lemma \ref{lem:fidelity}]
 We bound the difference
\begin{equation}
E_{\epsilon}(\phi) := \sup_{\rho\in[-1,1]}|\sup_{\rho_{0}\in \mathcal{T}_{\epsilon}(\phi)}\{h_{\rho_{0}}(\rho)\}- \sup_{\rho_{0}\in \mathcal{T}_{\epsilon}(\phi)}\{\log(p_{\gamma_{0},t_{0}}(\rho))\}|
\end{equation}
We break the analysis into three parts depending where $\rho$ belongs to.
The first case $\rho\in [-1+\delta,1-\delta]$ is the easier one, as due to Lemma \ref{lem:collision_basic}  and Corollary \ref{cor:sensitive_bounds} we have for all $\rho_{0}\in \mathcal{T}_{\epsilon}(\phi)$
\begin{equation}
-\log(C_{1})\leq \log(p_{\gamma_{0},t_{0}}(\rho)) - h_{\rho_{0}}(\rho) \leq \log C_{1} 
\end{equation}
Hence,
\begin{equation}
\sup_{\rho\in[-1+\delta,1-\delta]}|\sup_{\rho_{0}\in \mathcal{T}_{\epsilon}(\phi)}\{h_{\rho_{0}}(\rho)\}- \sup_{\rho_{0}\in \mathcal{T}_{\epsilon}(\phi)}\{\log(p_{\gamma_{0},t_{0}}(\rho))\}|\leq \log C_{1} 
\end{equation} 
We next treat the case $\rho\in [-1,-1+\delta]$. Recall that  $h_{\pm 1}(\rho):= -\frac{1\mp \rho}{1\pm \rho}\frac{t_{\pm 1}^{2}}{2} + \phi(\pm 1)$, 
where $t^{2}_{\pm 1} = 4\max\{\pm  \phi^{'}(\pm1), 0\}$. Assuming that $\phi$ is increasing at $-1$,  by construction $t^{2}_{-1} = 0$ and hence:
\begin{align}\label{eq:probs_lower1}
\sup_{\rho_{0}\in \mathcal{T}_{\epsilon}(\phi)}\{\log(p_{\gamma_{0},t_{0}}(\rho))\}, 
\sup_{\rho_{0}\in \mathcal{T}_{\epsilon}(\phi)}\{h_{\rho_{0}}(\rho)\}\geq h_{-1}(\rho)\geq  \phi(-1) 
\end{align}
Assuming that $\phi$ is decreasing at $-1$, we have $t_{-}^{2} =4|\phi^{'}(-1)|$ 
\begin{align}
\sup_{\rho_{0}\in \mathcal{T}_{\epsilon}(\phi)}\{h_{\rho_{0}}(\rho)\}&\geq h_{-1}(\rho) =  \phi(-1) - 2\frac{1+\rho}{1-\rho}|\phi^{'}(-1)|\geq \phi(-1)\label{eq:probs_lower2}
\end{align}
and by \eqref{eq:collision_right} in Lemma \ref{lem:collision_basic}  applied to $p_{+}(-\rho)$
\begin{align}
\sup_{\rho_{0}\in \mathcal{T}_{\frac{1}{2}}(\tilde{\phi})}\{\log(p_{\gamma_{0},t_{0}}(\rho))\} \geq- \frac{1}{2}\log C_{1}  -\frac{2}{2-\delta}\delta \phi'(-1)+\phi(-1)\geq -\frac{1}{2}\log C_{1} +\phi(-1)\label{eq:probs_lower3}
\end{align}

By Proposition \ref{prop:infimum} and the fact that $\phi$ can be written as the supremum of linear functions we get that $\sup_{\rho_{0}\in \mathcal{T}_{\epsilon}(\phi)}\{h_{\rho_{0}}(\rho)\}\leq \phi(\rho)$. Using Corollary \ref{cor:sensitive_bounds} and Corollary \ref{col:exponent_bounds}, we obtain:
\begin{align}
\sup_{\rho_{0}\in \mathcal{T}_{\epsilon}(\phi)\setminus\{-1,+1\}}\{p_{\gamma_{0},t_{0}}(\rho)\} &\leq -\frac{2-\delta}{\delta}\min_{\rho_{0}\in\mathcal{T}_{\epsilon}(\phi)\setminus\{-1,+1\}}\{t^{2}_{\gamma_{0}}\} +\frac{1}{2}\log C_{1} \\
&\leq -\frac{2-\delta}{\delta}\min_{\rho_{0}\in\mathcal{T}_{\epsilon}(\phi)\setminus\{-1,+1\}}\left\{- \frac{1+\rho_{0}^{2}}{1-\rho_{0}^{2}}\phi(\rho_{0})\right\}+\frac{1}{2}\log C_{1} 
\end{align}

To bound the above quantity further, distinguish two cases:   $\phi(-1)=0$ or $\phi(1)=0$. By convexity, in the former case  we have $\phi(\rho_{0})\leq \frac{1+\rho_{0}}{2} \phi(1)$ and $\phi(\rho_{0})\leq \frac{1-\rho_{0}}{2}\phi(-1)$ in the latter. Substituting these bounds and solving the optimization problem we find that the minimizer in the first case is $\rho_{0}=-\sqrt{2}+1$ and in the latter case $\rho_{0}=\sqrt{2}-1$. In both cases we may obtain:
\begin{align}
\sup_{\rho_{0}\in \mathcal{T}_{\frac{1}{2}}(\tilde{\phi})\setminus\{-1,+1\}}\{p_{\gamma_{0},t_{0}}(\rho)\} &\leq -\frac{2-\delta}{\delta}(\sqrt{2}-1)\max\{|\tilde{\phi}(1)|,|\tilde{\phi}(-1)|\} + \frac{1}{2}\log C_{1} 
\end{align}
Next, we obtain bounds for $\rho_{0}\in \{-1,+1\}$:
\begin{align}
\log(p_{-1}(\rho)) &\leq \phi(-1)\\
\log(p_{+1}(\rho)) &\leq \frac{1}{2}\log C_{1}  -2\frac{2-\delta}{\delta}\max\{\phi^{'}(1),0\}
\end{align}
Using the above inequalities we may conclude that:
\begin{align}
\sup_{\rho\in [-1,-1+\delta]}\sup_{\rho_{0}\in \mathcal{T}_{\epsilon}(\phi)}\{\log p_{\gamma_{0},t_{0}}(\rho)\} \leq \max\{\phi(-1),\frac{1}{2}\log C_{1} \}\leq \phi(-1)+\frac{1}{2}\log C_{1} \label{eq:probs_upper}
\end{align}
We have for $\rho\in [-1,-1+\delta]$ by \eqref{eq:probs_lower1} and \eqref{eq:probs_lower3}
\begin{align}
\sup_{\rho_{0}\in \mathcal{T}_{\epsilon}(\phi)}\{h_{\rho_{0}}(\rho)\}- \sup_{\rho_{0}\in \mathcal{T}_{\epsilon}(\phi)}\{\log(p_{\gamma_{0},t_{0}}(\rho))\} &\leq \phi(\rho) -\phi(-1)+\frac{1}{2}\log C_{1}  \leq  L(\phi)\delta +\frac{1}{2}\log C_{1} 
\end{align}
where in the last step we used the fact that $\phi$ is Lipischitz. In the same vein by \eqref{eq:probs_lower1} and \eqref{eq:probs_upper}
\begin{align}
\sup_{\rho_{0}\in \mathcal{T}_{\epsilon}(\phi)}\{h_{\rho_{0}}(\rho)\}- \sup_{\rho_{0}\in \mathcal{T}_{\epsilon}(\phi)}\{\log(p_{\gamma_{0},t_{0}}(\rho))\} & \geq \phi(-1) -\phi(-1)-\frac{1}{2}\log C_{1} = -\frac{1}{2}\log C_{1}
\end{align}
Using $\delta\leq \frac{\epsilon}{L(\phi)} \Rightarrow L(\phi)\delta \leq \epsilon \leq \epsilon \log C_{1}$.  
By symmetry the case $\rho\in [1-\delta,1]$ follows. Overall, for $\epsilon=1/2$ we obtain the   bound $E_{1/2}(\phi) \leq\log C_{1}$.
\end{proof}

\begin{proof}[Proof of Lemma \ref{lem:complexity}]
Let $\delta_{1/2}:= \frac{k^{*}}{2\beta L(\phi)}$  be the constant from Lemma \ref{lem:boundary} applied for $\tilde{\phi}$, then  for all $\rho_{0}\in \mathcal{T}_{\frac{1}{2}}(\tilde{\phi})\setminus \{-1,+1\}$ we have $|\rho_{0}|\leq 1-\delta_{1/2}$.  Using $\tilde{\phi}(\rho_{0})\leq R(\tilde{\phi}) = \frac{\beta}{k}R(\phi)$ and $|\rho_{0}|\leq 1-\delta_{1/2}=1-\frac{k}{2\beta L(\phi)}$ for $\rho_{0}\neq \pm 1$, we get by Corollary \ref{col:exponent_bounds} that $\sup_{\rho_{0}\in \mathcal{T}_{\frac{1}{2}}(\tilde{\phi})}t_{\gamma_{0}}^{2} \leq  8 L(\tilde{\phi})R(\tilde{\phi}) \leq 8\left(\frac{\beta}{k} \right)^{2} L(\phi) R(\phi)$. For $\rho_{0}\in \{-1,1\}$ we have $t^{2}\leq 4|\phi'(\rho)|\leq 4 L(\tilde{\phi})\leq 8 L(\tilde{\phi})R(\tilde{\phi})$ for $R(\tilde{\phi})\geq 1/2$.
\end{proof}
\section{Open Questions} \label{sec:open}
\phantom{p}

\paragraph{Data-dependent LSH} Both the HBE and Multi-Resolution HBE approaches exhibit $1/\sqrt{\mu}$ complexity depending on $\mu=Z_{w}(y)$. For HBE~\cite{charikar2017hashing},  the instance that instantiates the worst-case variance fo the estimator is when there are $O(n\mu)$ points very close to the query such that $w(x_{1},y)=\Theta(1)$ and $O(n)$ points ``away" from the query such that $w(x_{2},y)=\Theta(\mu)$. On the other hand for MR-HBE, if one uses the full power of Theorem \ref{thm:psquared} (see Section \ref{sec:variance})  by analyzing $D_{T}(x_{1},x_{2})$ rather than its simplified version Theorem \ref{thm:scale-free}, the worst case instance for the variance appears to have $O(n\sqrt{\mu})$ points with $w(x_{1},y)=\Theta(\sqrt{\mu})$ and $O(n)$ points with $w(x_{2},y)=\Theta(\mu)$. For the Gaussian kernel this essentially  means that it involves solving a $c$-ANN problem with $c=\sqrt{2}$. Using the best data-independent LSH~\cite{andoni2006near}  the running time should be $n^{1/c^{2}+o(1)}=n^{1/2+o(1)}$ matching the $1/\sqrt{\mu}$ dependence exhibited by our data structures. This suggests that if one is able to adapt the data-dependent hashing approach~\cite{andoni2015optimal,andoni2017optimal} to this setting one might be able to get algorithms running in $n^{1/(2c^{2}-1)}=n^{1/3+o(1)}$ time or $1/\sqrt[3]{\mu}$. We believe this is an intriguing direction for future work.

\paragraph{Cell-probe Lower bounds} The batch version of the problem, where we seek to answer many queries, is equivalent to approximating a matrix-vector product. The matrix in question has elements given by  $w(x,y)$ for $x\in X$ and $y\in Y$. In   high dimensions for fast decaying functions like the Gaussian,  this problem is related to Boolean Matrix Vector Multiplication. For the latter problem and succinct data-structures recently \cite{chakraborty2018tight} a tight cell-prove lower bound of $\tilde{O}(n^{3/2})$ was given. This matches the $n^{3/2+o(1)}$ complexity of our data- structures for $L\leq \frac{1}{2}\log n$.  

The lower bound is based on the fact that there is a distribution over boolean matrices where  querying arbitrary elements of the matrix does not reveal too much information and there is a set of vector query  whose answer reveals a large amount of information about the matrix. This is used to show that any succinct data-structure that can answer the queries without reading too many elements from the matrix must have stored a lot of information.  The parallel to our case would be that ``vector" queries specify a subset of points in our data set $X$ for which we want to know the density for a fixed set of queries $Y$, and ``element" queries correspond to evaluating the value $w(x,y)$ between a query and a point. HBE essentially define  data-structures using $n^{3/2+o(1)}$ extra bits of storage  that for a single  vector (e.g. all ones) one can answer $n$ ``point" queries  using $n^{1/2+o(1)}$ evaluations of $w(x,y)$ per query point $y$. Formalizing this connection is an interesting research question.

\paragraph{Locality Sensitive Hashing}  One disadvantage of many LSH based approaches is that hash functions  often   can be expensive to compute at least in the form suggested by the theory. In recent years there has been an effort to design practical hash functions that come close to the performance of the optimal ones. For example the papers~\cite{kennedy2017fast,andoni2015practical} study practical functions for the unit sphere, while \cite{andoni2017lsh} study functions for the binary hypercube. Combining these novel LSH methods with the method of Hashing Based Estimators introduced  in \cite{charikar2017hashing} and extended here, is a promising direction to getting practical algorithms for  estimation problems.

\paragraph{Variance Reduction} The topic of Variance Reduction for Stochastic Gradient~\cite{le2012stochastic,johnson2013accelerating,shalev2013stochastic} is an important field of current research. There are roughly three almost orthogonal approaches to this problem: re-weighting schemes~\cite{lan2012optimal,allen2016variance,allen2017katyusha}, importance sampling schemes~\cite{zhao2015stochastic,allen2016even} and partition-based schemes~\cite{zhao2014accelerating,allen2016exploiting}. For almost all these approaches, the distribution that gradients are sampled is independent of the current iterate (e.g. uniform or based on Lipschitz constants of gradients),  or changes with the current iterate and requires linear time to update the new distributions. The latter approaches are referred to as Adaptive Variance Reduction methods~\cite{csiba2015stochastic,namkoong2017adaptive,salehi2017stochastic}. Our approach sidesteps the issue of recomputing such distributions through the use of Locality Sensitive Hashing. An intriguing direction is to utilize our techniques within an optimization algorithm to obtain faster optimization methods.

\section*{Acknowledgments}
The authors would like to thank Dimitris Achlioptas and Clement Canonne for valuable feedback on improving the presentation of the paper, as well as Aviad Rubinstein for helpful conversations on conditional lower bounds.  We are also grateful to Casper Freksen for pointing out a number of typos on an earlier version of the paper.  The second author is partially supported by a Onassis Foundation Scholarship.	
\bibliography{kde}
\bibliographystyle{abbrv}
\end{document}